\documentclass{article} % For LaTeX2e
\usepackage{iclr2026_conference,times}

\usepackage{amsmath,amsfonts,bm}

\def\eqref#1{equation~\ref{#1}}
\def\1{\bm{1}}

\DeclareMathAlphabet{\mathsfit}{\encodingdefault}{\sfdefault}{m}{sl}
\SetMathAlphabet{\mathsfit}{bold}{\encodingdefault}{\sfdefault}{bx}{n}

\newcommand{\KL}{D_{\mathrm{KL}}}

\usepackage{hyperref}
\usepackage{url}
\usepackage{amsmath,amssymb,amsthm}
\usepackage{graphicx}
\usepackage{booktabs}
\usepackage{tikz}
\usetikzlibrary{arrows.meta,positioning}
\usepackage{quantikz}

\newtheorem{theorem}{Theorem}
\newtheorem{lemma}{Lemma}
\newtheorem{proposition}{Proposition}
\newtheorem{assumption}{Assumption}
\newtheorem{definition}{Definition}
\newtheorem{remark}{Remark}

\newcommand{\Deff}{d_{\mathrm{eff}}}
\providecommand{\KL}{\mathrm{KL}}
\renewcommand{\KL}{\mathrm{KL}}
\newcommand{\Fmat}{\hat{\mathcal{F}}}

\title{Entanglement as a Structural Complexity Axis:\\ A PAC-Bayesian View of Generalization in\\ Quantum Policies and Value Functions}

\author{
Jian Xu$^{1,2}$, \quad Delu Zeng$^{3}$, \quad John Paisley$^{4}$, \quad Qibin Zhao$^{2}$ \\[3pt]
\normalfont\normalsize
$^{1}$RIKEN iTHEMS \quad $^{2}$RIKEN AIP \quad $^{4}$Columbia University \\
\normalfont\normalsize
$^{3}$South China University of Technology \quad \texttt{jian.xu@riken.jp}
}

\iclrfinalcopy
\begin{document}

\maketitle
\lhead{}\rhead{}\chead{}\renewcommand{\headrulewidth}{0pt}  % strip the ICLR conference header/rule for the preprint

\begin{abstract}
Parameterized quantum circuits (PQCs) are increasingly used as policies and value
functions in quantum reinforcement learning (QRL), yet almost all evaluations report
only average return, leaving open the question of \emph{when} and \emph{why} a quantum
policy generalizes. We give a PAC-Bayesian answer. We derive a generalization bound for
stochastic quantum policies whose complexity term is controlled not by the raw number of
circuit parameters, but by the \emph{effective dimension of the Fisher geometry induced by
the circuit}---a quantity that entanglement inflates. Empirically, in a controlled
setting that fixes the number of trainable rotations and varies \emph{only} the
entangling connectivity, we find that entanglement is an independent axis of complexity:
at a fixed parameter count the train-to-test generalization gap grows with the circuit's
Fisher effective dimension---which entanglement inflates---whereas raw parameter count is the
weakest predictor of the gap. The bound is primarily a \emph{ranking certificate}: it
correctly orders circuits of identical parameter count, something a parameter-counting bound
cannot do at all. We confirm this mechanism across settings---supervised classification (at $4$--$16$ qubits, on synthetic and on real Iris/Breast-Cancer/Wine data), a reward-only quantum
contextual bandit, and multi-step value-function generalization---where entangled circuits
generalize worse than non-entangled ones of identical parameter count and the gap shrinks
with sample size as the bound predicts. Our strongest evidence is in these low-variance
decision models (single-observable classifiers and value heads, and a one-step policy); in
genuine end-to-end multi-step \emph{policy} learning the standard entangler's effect is
statistically significant but return variance leaves the full ordering only partially
resolved, so we frame the paper as a controlled study of the entanglement--generalization
trade-off in quantum decision models rather than a solved account of end-to-end quantum RL.
A partial-correlation analysis shows the Fisher
effective dimension screens off the entangling pattern (which adds $\Delta R^2\!<\!0.01$ once
$\Deff$ is controlled), controls (matched training accuracy, alternative readout and
optimizer) rule out optimization confounders, and the effect survives execution on an IBM
Heron quantum processor under real noise. Our results reframe the design of quantum policies around an
entanglement--generalization trade-off rather than around expressivity alone.
\end{abstract}

\section{Introduction}

Quantum reinforcement learning (QRL) replaces the neural policy or value network of a
classical agent with a parameterized quantum circuit (PQC). Recent work has made such
hybrid agents practical---PPO-Q integrates a PQC policy into proximal policy optimization
and validates it on superconducting hardware \citep{jin2025ppo}, and adaptive non-local
observables show that the measurement layer itself is a key design axis
\citep{lin2025quantum}. Yet the field is almost entirely empirical: papers report the average
return of a chosen ansatz, and a reviewer is left to wonder whether a reported gain
reflects a genuine inductive bias or merely the implicit regularization of a small model.

We argue that the right question is not ``does a quantum policy achieve higher return''
but ``\emph{when does a quantum policy generalize, and what property of the circuit
controls it}''. This is the question PAC-Bayes theory is built to answer for stochastic
predictors \citep{mcallester1999pac,dziugaite2017computing}, and it has recently been brought to
reinforcement learning, yielding non-vacuous certificates for modern off-policy agents
\citep{zitouni2025pac}. In parallel, a PAC-Bayesian analysis of \emph{supervised} quantum
models has appeared \citep{rodriguez2026pac}, whose complexity term is built from learned
parameter norms and sparsity. To our knowledge no PAC-Bayes analysis addresses quantum
\emph{policies}, and none isolates \emph{entanglement} as the governing complexity.

This paper makes three contributions.
\begin{itemize}
\item \textbf{A PAC-Bayes bound for quantum policies} (Section~\ref{sec:theory}) whose
complexity term is the effective dimension of the circuit-induced Fisher geometry rather
than the parameter count. We show that under a data-dependent Gaussian posterior the KL
complexity is governed by $\Deff(\theta)$, which for PQCs is inflated by entanglement.
\item \textbf{An identification of entanglement as an independent complexity axis}
(Section~\ref{sec:decouple}). Fixing the number of trainable rotations and varying only
the entangling connectivity, we show the generalization gap grows monotonically with the
Meyer--Wallach entangling power \emph{at fixed parameter count}, while parameter count
itself is the weakest predictor of the gap. The axis is \emph{readout-gated}: the mechanism
is that entanglement enlarges the readout's light-cone, so it is decisive under the local
single-qubit readouts standard in QRL and is partly substituted by a global readout, which
activates the same parameters directly (Section~\ref{sec:controls})---a control that directly
confirms the light-cone account and shows $\Deff$, not the entangling label, is what tracks
the gap.
\item \textbf{Broad empirical validation}: supervised classification at $4$/$6$/$8$ qubits
(Sections~\ref{sec:decouple}--\ref{sec:supervised}), a reward-only quantum contextual bandit
and multi-step value-function generalization (Section~\ref{sec:rl}), optimization controls,
and a run on IBM Heron hardware (Section~\ref{sec:controls})---in all of which entangled
circuits generalize worse than non-entangled ones of identical parameter count.
\end{itemize}

The practical message is a design principle: entanglement is a double-edged sword. It
enlarges the function space a quantum policy can represent, but it carries a
generalization cost that a parameter-counting analysis does not see and that our bound
makes explicit.

\paragraph{Scope.} We are deliberate about what we claim. The theory is stated for policies
and value functions, and our \emph{strongest} evidence comes from the low-variance
instantiations of exactly these objects: single-observable classifiers and value heads, a
one-step (contextual-bandit) policy, and Monte-Carlo evaluation of a multi-step value
function. In genuine end-to-end multi-step \emph{policy} learning (REINFORCE) the standard
entangler's effect \emph{is} statistically significant (Section~\ref{sec:rl}), but return-based
estimates are high enough variance that the full ordering is not always resolved; we therefore
treat the low-variance settings as the primary probes. The paper is thus best read as a study
of the entanglement--generalization trade-off in quantum decision models---policies and value
functions---rather than a claim to have solved generalization in end-to-end quantum RL.

\section{Related Work}

\paragraph{Quantum reinforcement learning.} PQC agents have been trained with policy
gradients, actor--critic and value-based methods
\citep{jin2025ppo,luo2025training,wu2025quantum}, with recent attention to the
measurement layer \citep{lin2025quantum} and to data re-uploading and trainability
\citep{coelho2024vqc}. These works optimize and report return; generalization is
not analyzed.

\paragraph{Generalization of quantum models.} A now-substantial literature bounds the
generalization of quantum models. Covering-number and Rademacher analyses give uniform
bounds for quantum feature maps and PQCs \citep{banchi2021generalization}, and
\citet{caro2022generalization} prove the landmark result that a PQC with $T$ parameterized
gates generalizes from $\widetilde O(T)$ training points---a bound driven by \emph{gate
count}, not entanglement. \citet{gilfuster2024rethinking} caution that such
uniform-convergence bounds can be simultaneously satisfied and uninformative for the models
one actually trains, motivating \emph{model-dependent} capacity measures. The
Fisher-information ``effective dimension'' is one such measure, proposed as a capacity that
reflects the trained model rather than the worst case \citep{abbas2021power}. Our $\Deff$ is
this same log-determinant functional, but our contribution is orthogonal to
\citet{abbas2021power} in two respects (Appendix~\ref{app:extra} makes the comparison
precise): (i) they use $\Deff$ as a global, sample-size-dependent capacity for a model
\emph{class}, whereas we evaluate it at the \emph{trained} parameters with a fixed $\gamma$
and use it only to \emph{rank} circuits of equal parameter count; and (ii) they do not
connect $\Deff$ to entanglement---our Proposition~\ref{prop:ent} \emph{attributes} the
inflation of the Fisher rank (and hence of $\Deff$) to the entanglement-driven growth of the
readout light-cone, which is the new mechanism. Most closely related on the PAC-Bayes side,
\citet{rodriguez2026pac} give a PAC-Bayes bound for \emph{supervised} quantum classifiers
whose complexity is a parameter-norm/sparsity term; they report weak but positive
correlations between their complexity and the observed gap. We differ in two ways: we treat
\emph{policies} (return-based generalization under Markov dependence) and we identify
\emph{entanglement}, not parameter norm or gate count, as the governing axis, isolating it
at fixed parameter and gate count.

\paragraph{PAC-Bayes for RL and control.} PAC-Bayes control learns policies with
certified generalization to novel environments \citep{majumdar2021pac}, and recent
work derives PAC-Bayes RL bounds that account for the mixing time of the induced Markov
chain \citep{zitouni2025pac}. We adopt this template and instantiate its complexity term for
quantum policies.

\section{A PAC-Bayes Bound for Quantum Policies}
\label{sec:theory}

\paragraph{Setup.} An agent interacts with an MDP $M$ drawn from a distribution
$\mathcal{D}$ over environments. A stochastic quantum policy $\pi_\theta(a\mid s)$ is
defined by a PQC $U(s,\theta)$ acting on $n$ qubits: the state $s$ is encoded, the circuit
is applied, and the action is read from a \emph{single} readout observable
$\langle Z_0\rangle_{s,\theta}$ (a binary policy $\pi_\theta(a{=}1\mid s)=\sigma(\alpha\langle Z_0\rangle_{s,\theta})$;
$k$-action policies use $k$ observables and our analysis applies per observable). We use the
single-observable readout throughout, both in the theory below and in every experiment. We
write $J(\pi_\theta)=\mathbb{E}_{M\sim\mathcal{D}}
\mathbb{E}_{\tau\sim\pi_\theta}[R(\tau)]$ for the expected return with $R\in[0,R_{\max}]$,
and $\hat J(\pi_\theta)$ for its empirical estimate on $m$ training environments. We
consider a Gaussian posterior $Q=\mathcal{N}(\theta,\sigma^2 I_d)$ over the $d$ circuit
parameters and a prior $P=\mathcal{N}(\theta_0,\sigma^2 I_d)$.

\begin{theorem}[PAC-Bayes bound for quantum policies]
\label{thm:main}
Under the mixing-time assumptions of \citet{zitouni2025pac}, for any $\delta\in(0,1)$, with
probability at least $1-\delta$ over the draw of $m$ training environments, every
posterior $Q$ satisfies
\begin{equation}
J(\pi_Q) \;\ge\; \hat J(\pi_Q) \;-\; R_{\max}\,\sqrt{\frac{\KL(Q\Vert P) + \ln\frac{2\sqrt{m}}{\delta}}{2\,m_{\mathrm{eff}}}},
\qquad
\KL(Q\Vert P)=\frac{\lVert \theta-\theta_0\rVert^2}{2\sigma^2},
\label{eq:bound}
\end{equation}
where $m_{\mathrm{eff}}=m/\kappa$ discounts the sample size by the chain mixing factor
$\kappa$.
\end{theorem}

Equation~\eqref{eq:bound} is the standard McAllester bound with the RL sample correction;
its only quantum-specific object is $\KL(Q\Vert P)$. Naively this depends only on the
Euclidean displacement of the $d$ parameters, suggesting that generalization is controlled
by parameter count. Our central observation is that this is misleading once the posterior
is allowed to be shaped by the local geometry the circuit induces on the loss.

\paragraph{Environments, contexts, and the single-step case.} The i.i.d.\ sample $z_i$ of
Theorem~\ref{thm:main} is an ``environment'' $M_i\sim\mathcal{D}$; the mixing factor $\kappa$
enters \emph{only} because a multi-step return pools Markov-dependent transitions. In the
single-step settings that carry our cleanest evidence---supervised classification and the
contextual bandit---each ``environment'' is a single i.i.d.\ \emph{context} $x_i\sim\mathcal{D}$
with a one-shot reward, the trajectory has length one, so there is no Markov dependence to
discount and $\kappa=1$, giving $m_{\mathrm{eff}}=m=N$ exactly. The bound we actually
evaluate in these sections is therefore the plain McAllester form, and $\kappa$ is relevant
only to the genuine multi-step experiments (value-function regression and end-to-end
REINFORCE, Section~\ref{sec:rl}). We use ``environment'' and ``context'' interchangeably in
the single-step case for this reason. The distribution shift the experiments measure is
thus the in-distribution train$\to$test gap over i.i.d.\ contexts, which is exactly the
object Theorem~\ref{thm:main} bounds when $m$ contexts are drawn i.i.d.; the meta/multi-task
reading (many environments) is the same inequality specialized to trajectory length one.

\paragraph{From an isotropic to a Fisher-shaped posterior.} The bound~\eqref{eq:bound} is
loose because the isotropic posterior $\mathcal{N}(\theta,\sigma^2 I)$ ignores the geometry
of the loss. The curvature of a quantum policy is the Fisher information matrix
$\Fmat(\theta)=\mathbb{E}_{s}[\nabla_\theta \ell_{s}\,\nabla_\theta \ell_{s}^\top]$, with
eigenvalues $\lambda_1\ge\dots\ge\lambda_d\ge 0$, and it is highly anisotropic. For the
single-observable Bernoulli readout $p_s=\sigma(\alpha\langle Z_0\rangle_{s,\theta})$ used
throughout, this loss Fisher has the explicit closed form
$\Fmat(\theta)=\mathbb{E}_{s}\!\big[\alpha^2\,p_s(1{-}p_s)\,\nabla_\theta f_s\,\nabla_\theta f_s^\top\big]$,
with $f_s=\langle Z_0\rangle_{s,\theta}$---i.e.\ the loss Fisher of the logistic head
\emph{equals} the output-gradient Fisher weighted by the Bernoulli variance $\alpha^2 p_s(1{-}p_s)$.
This is exactly the object we diagonalize in every experiment (Appendix~\ref{app:extra}), so
the $\Fmat$ appearing in the bound and the $\Fmat$ we measure are the same matrix, and its
rank is governed by the readout light-cone (Proposition~\ref{prop:ent}). Choosing a
posterior that respects this curvature turns the KL term into a genuinely quantum quantity.

\begin{theorem}[Fisher-form bound]
\label{thm:fisher}
Let the prior be the data-independent isotropic Gaussian $P=\mathcal{N}(\theta_0,\tau^2 I_d)$
and the posterior the Gauss--Newton Gaussian $Q=\mathcal{N}(\theta,\Sigma)$ with
$\Sigma=\tau^2(I_d+\gamma\Fmat)^{-1}$, $\gamma>0$. Then
\begin{equation}
\KL(Q\Vert P)=\frac{\lVert\theta-\theta_0\rVert^2}{2\tau^2}
+\tfrac12\,\underbrace{\log\det\!\big(I_d+\gamma\Fmat\big)}_{=:\ \Deff(\gamma)}
-\tfrac12\sum_{i=1}^d\frac{\gamma\lambda_i}{1+\gamma\lambda_i},
\label{eq:klfisher}
\end{equation}
and hence, under the assumptions of Theorem~\ref{thm:main}, with probability $\ge1-\delta$,
\begin{equation}
J(\pi_Q)\ \ge\ \hat J(\pi_Q)-R_{\max}\sqrt{\frac{\dfrac{\lVert\theta-\theta_0\rVert^2}{2\tau^2}
+\tfrac12\,\Deff(\gamma)+\ln\frac{2\sqrt m}{\delta}}{2\,m_{\mathrm{eff}}}}.
\label{eq:boundfisher}
\end{equation}
\end{theorem}

Equation~\eqref{eq:boundfisher} is the bound that matters: its complexity is the distance
term \emph{plus} the \emph{Fisher effective dimension}
\begin{equation}
\Deff(\gamma)=\log\det\!\big(I_d+\gamma\Fmat\big)=\sum_{i=1}^d\log(1+\gamma\lambda_i),
\label{eq:deff}
\end{equation}
a quantity a parameter count cannot see. The final term of~\eqref{eq:klfisher} is
\emph{subtracted} and lies in $[0,d/2]$; dropping it therefore replaces $\KL$ by a larger
upper surrogate, which conservatively \emph{loosens} the gap bound while keeping it
valid---we drop it only to expose $\Deff$ as the operative complexity. Two facts make
$\Deff$ the right
object: $\Deff(\gamma)\le d$ always, and $\Deff(\gamma)\approx\mathrm{rank}(\Fmat)\log\gamma$
for large $\gamma$, so it counts \emph{active} Fisher directions rather than nominal
parameters. Crucially the prior is data-\emph{independent} (centred at the initialization
$\theta_0$), so~\eqref{eq:boundfisher} is a valid PAC-Bayes certificate even though the
posterior \emph{shape} $\Sigma$ is data-dependent through $\Fmat$; the participation ratio
$(\sum_i\lambda_i)^2/\sum_i\lambda_i^2$ we also report is a stable empirical surrogate for
the same spectrum, not a separate capacity notion. Equation~\eqref{eq:klfisher} is a direct
Gaussian KL computation, proved in Appendix~\ref{app:proofs}.

\begin{remark}[Choosing $\gamma$]
\label{rem:gamma}
The scale $\gamma=\beta\tau^2$ is a free posterior hyperparameter and may be optimized. Since
it must be fixed before seeing the data to keep the prior valid, one applies a union bound
over a geometric grid $\gamma\in\{2^{j}\}_{j=1}^{J}$, replacing $\delta$ by $\delta/J$ and
paying only an additive $\tfrac12\ln J$ inside the root; taking $J=O(\log m)$ and then the
best grid point costs a negligible $O(\log\log m)$ term and lets the data pick the $\gamma$
that minimizes the bound. In every experiment in this paper we fix a single
\emph{constant} $\gamma=50$---the same value for all circuits and sample sizes---so that
$\Deff$ is comparable across configurations; Appendix~\ref{app:extra} gives the full Fisher
computation. (A sample-size-dependent $\gamma\!\propto\!m$, as in \citet{abbas2021power},
would rescale $\Deff$ per row and is \emph{not} used here.)
\end{remark}

\begin{proposition}[Entanglement raises the effective dimension]
\label{prop:ent}
Fix the number and placement of the trainable single-qubit rotations and let the entangling
connectivity vary. Let $\mathcal{L}(\theta)\subseteq\{1,\dots,d\}$ be the set of parameters
inside the backward light-cone of the readout observable (Def.~\ref{def:cone}). Then
generically $\mathrm{rank}(\Fmat)=|\mathcal{L}(\theta)|$, and $|\mathcal{L}(\theta)|$---and
therefore $\Deff(\gamma)$---is non-decreasing under the addition of entangling gates, while
the parameter count $d$ is unchanged. With no entangling gates $\mathcal{L}=\{i:q(i)=0\}$,
so $\Deff$ is at its minimum; any connectivity that couples the readout wire to others
strictly enlarges $\mathcal{L}$ and raises $\Deff$.
\end{proposition}

Proposition~\ref{prop:ent} yields a sharp, falsifiable prediction we test at fixed parameter
count: entangled circuits have a larger $\Deff$, hence a looser bound~\eqref{eq:boundfisher},
hence a larger generalization gap, than non-entangled circuits with \emph{identical}
parameter count---something no parameter-counting bound can express. Its rigorous proof (a
light-cone lemma for the rank monotonicity, plus a genericity assumption for the rank
identity) is in Appendix~\ref{app:proofs}. The prediction holds in the strong monotone form
(none $<$ linear $<$ full in both $\Deff$ and gap) in supervised classification
(Sections~\ref{sec:decouple}--\ref{sec:supervised}) and, at the smaller sample sizes, in the
RL bandit, where each entangler also beats \texttt{none} by more than three standard errors
(Section~\ref{sec:rl}). Figure~\ref{fig:mech} summarizes the mechanism.

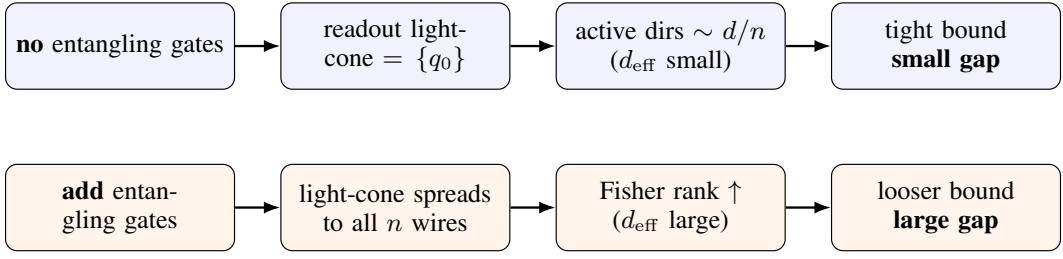
\begin{figure}[t]
\centering
\resizebox{\textwidth}{!}{%
\begin{tikzpicture}[
  box/.style={draw,rounded corners,align=center,minimum height=1.05cm,text width=2.55cm,font=\footnotesize},
  arr/.style={-{Latex[length=2mm]},thick}]
\node[box,fill=blue!5] (a1) {\textbf{no} entangling gates};
\node[box,right=0.55cm of a1,fill=blue!5] (a2) {readout light-cone $=\{q_0\}$};
\node[box,right=0.55cm of a2,fill=blue!5] (a3) {active dirs $\sim d/n$ \\ ($\Deff$ small)};
\node[box,right=0.55cm of a3,fill=blue!5] (a4) {tight bound \\ \textbf{small gap}};
\node[box,below=0.9cm of a1,fill=orange!8] (b1) {\textbf{add} entangling gates};
\node[box,right=0.55cm of b1,fill=orange!8] (b2) {light-cone spreads to all $n$ wires};
\node[box,right=0.55cm of b2,fill=orange!8] (b3) {Fisher rank $\uparrow$ \\ ($\Deff$ large)};
\node[box,right=0.55cm of b3,fill=orange!8] (b4) {looser bound \\ \textbf{large gap}};
\foreach \i/\j in {a1/a2,a2/a3,a3/a4,b1/b2,b2/b3,b3/b4} \draw[arr] (\i)--(\j);
\end{tikzpicture}}
\caption{The mechanism. At a fixed parameter count, entangling gates enlarge the backward
light-cone of the readout observable (Lemma~\ref{lem:lightcone}), which raises the rank and
hence the effective dimension $\Deff=\log\det(I+\gamma\Fmat)$ of the Fisher geometry
(Proposition~\ref{prop:ent}). A larger $\Deff$ enlarges the complexity term of the
PAC-Bayes bound~\eqref{eq:boundfisher} and thus the generalization gap---an effect invisible
to a parameter-counting analysis, since $d$ is unchanged.}
\label{fig:mech}
\end{figure}

\section{Entanglement is an Independent Axis of Complexity}
\label{sec:decouple}

\paragraph{Design.} To separate entanglement from parameter count we use an ansatz with a
fixed layout of trainable rotations---$L$ layers of $R_Y,R_Z$ on each of $n{=}4$ qubits,
so $d=2nL$ parameters---and vary \emph{only} the entangling connectivity inserted after
each rotation layer: \texttt{none} (no CNOTs), \texttt{linear} (a CNOT chain), and
\texttt{full} (all-to-all CNOTs). All three share exactly the same $d$; they differ only
in entangling power. The binary target $y=\mathrm{sign}(x_1x_2-x_3x_4+\tfrac12 x_1)$ requires
feature interactions, inputs are angle-encoded, the classifier is the single-observable
readout $\hat y=\mathrm{sign}(\langle Z_0\rangle)$, and we sweep the training-set size
$N$ and evaluate on $2000$ held-out inputs, averaging $8$ seeds per cell.

\paragraph{Result: entangled circuits genuinely generalize worse, not just memorize more.}
A key check is that the models actually generalize---otherwise a larger gap would only mean
more memorization. Table~\ref{tab:decouple} isolates depth $L{=}2$ ($d{=}16$) at two sample
sizes. At the small $N{=}16$ the entangled circuit does memorize (train $0.92$ vs $0.69$),
but this is exactly the ambiguous regime a reviewer should worry about: test accuracy is
near chance for all three. At $N{=}64$, however, \emph{every} circuit generalizes above the
$0.5$ chance level (test $0.61$--$0.66$), and there the entangled circuits are strictly
\emph{worse}: the all-to-all circuit reaches test accuracy $0.608$ versus $0.656$ for the
non-entangled one---which itself has essentially zero gap---and the gap is ordered
none $<$ linear $<$ full at every $N$ (Figure~\ref{fig:main}, left).
The effect is depth-gated (a single entangling layer, $L{=}1$, is inert for this
single-qubit readout) and holds at $d$ fixed throughout. A caveat on small-$N$ magnitudes:
the slightly \emph{negative} \texttt{none} gap at $N{=}64$ ($-0.01$) is within one seed
standard error of zero---at small $N$ the finite-sample noise in the gap estimate is
comparable to the effect for the (near-zero-gap) unentangled circuit---so we anchor every
quantitative contrast at the sample sizes where the gaps are resolved beyond their standard
error (the $\ge3\sigma$ \texttt{full}$-$\texttt{none} separations in Tables~\ref{tab:rl}
and~\ref{tab:scale}) and read the small-$N$, near-chance rows only as tests of the
\emph{ordering}, not of gap magnitudes.

\begin{table}[t]
\centering
\caption{Fixed parameter count ($d{=}16$, $L{=}2$), varying only the entangling connectivity,
single-observable readout, $8$ seeds. At $N{=}64$ all three circuits generalize well above
the $0.5$ chance level, yet the entangled circuits have \emph{lower test accuracy} and a
larger gap---genuinely worse generalization, not mere memorization.}
\label{tab:decouple}
\begin{tabular}{llcccc}
\toprule
$N_{\mathrm{train}}$ & Entanglement & $Q$ (M.--W.) & train acc & test acc & gap \\
\midrule
$16$ & \texttt{none}   & $0.00$ & $0.69$ & $0.59$ & $0.101$ \\
$16$ & \texttt{linear} & $0.45$ & $0.75$ & $0.60$ & $0.146$ \\
$16$ & \texttt{full}   & $0.56$ & $0.92$ & $0.54$ & $0.381$ \\
\midrule
$64$ & \texttt{none}   & $0.00$ & $0.65$ & $\mathbf{0.656}$ & $-0.01$ \\
$64$ & \texttt{linear} & $0.51$ & $0.72$ & $0.657$ & $0.064$ \\
$64$ & \texttt{full}   & $0.58$ & $0.71$ & $\mathbf{0.608}$ & $0.104$ \\
\bottomrule
\end{tabular}
\end{table}

\paragraph{Result: $\Deff$ is the best---and significantly best---predictor.} Across a large
grid of $300$ configurations (five entangling patterns $\times$ five depths $\times$ three
sample sizes $\times$ four seeds; Appendix~\ref{app:extra}), the Spearman correlation with the
gap is $\rho{=}0.82$ for the Fisher effective dimension $\Deff{=}\log\det(I{+}\gamma\Fmat)$,
$0.73$ for the learned displacement $\lVert\theta-\theta_0\rVert^2$, $0.72$ for the entangling
power $Q$, and $0.45$ for parameter count. The margin of $\Deff$ over the learned-norm term is
\emph{bootstrap-significant}: $\Delta\rho=0.090$ with $95\%$ CI $[0.041,0.144]$. Thus the
Fisher effective dimension---the quantity that actually appears in
bound~\eqref{eq:boundfisher}---is not merely nominally but significantly the best predictor of
the gap, ahead of the learned-norm term used by \citet{rodriguez2026pac} and far ahead of raw
parameter count; and, unlike the norm, entanglement is a structural property fixed before
training that separates circuits of \emph{equal} parameter count.

\section{The Fisher Effective Dimension is the Right Complexity}
\label{sec:supervised}

We now test Proposition~\ref{prop:ent} directly by computing, for each circuit, the Fisher
effective dimension of Eq.~\eqref{eq:deff} at the trained parameters and asking whether it
(a) tracks the gap, (b) beats parameter count, and (c) is inflated by entanglement at
fixed $d$.

\begin{table}[t]
\centering
\caption{Complexity measures vs.\ generalization gap: Spearman $\rho$ across a $300$-config
grid ($5$ entangling patterns $\times$ $5$ depths $\times$ $3$ sample sizes $\times$ $4$
seeds). The Fisher effective dimension $\Deff=\log\det(I+\gamma\Fmat)$ (the quantity in the
bound) is the strongest predictor; its margin over the learned norm is bootstrap-significant
($\Delta\rho{=}0.090$, $95\%$ CI $[0.041,0.144]$). Raw parameter count is the weakest.}
\label{tab:fisher}
\begin{tabular}{lcccc}
\toprule
measure & param count & Meyer--Wallach $Q$ & $\lVert\theta-\theta_0\rVert^2$ & Fisher $\Deff$ (log-det) \\
\midrule
$\rho(\cdot,\text{gap})$ & $0.45$ & $0.72$ & $0.73$ & $\mathbf{0.82}$ \\
\bottomrule
\end{tabular}
\end{table}

At fixed depth (fixed $d$), moving from \texttt{none} through \texttt{linear} to
\texttt{full} connectivity inflates the Fisher effective dimension
$\Deff=\log\det(I+\gamma\Fmat)$ monotonically, from $\Deff\!\approx\!3.3$ (none) to $9.3$
(linear) to $21.8$ (full) at $L{=}2$, confirming that entanglement raises the effective
complexity that Eq.~\eqref{eq:boundfisher} charges for while $d$ is held fixed. The
generalization gap also shrinks with the training-set size $N$
(Figure~\ref{fig:main}, left), a basic PAC-Bayes consistency check, and it increases with
the Fisher effective dimension across circuits (Figure~\ref{fig:main}, right).

\begin{figure}[t]
\centering
\includegraphics[width=0.48\textwidth]{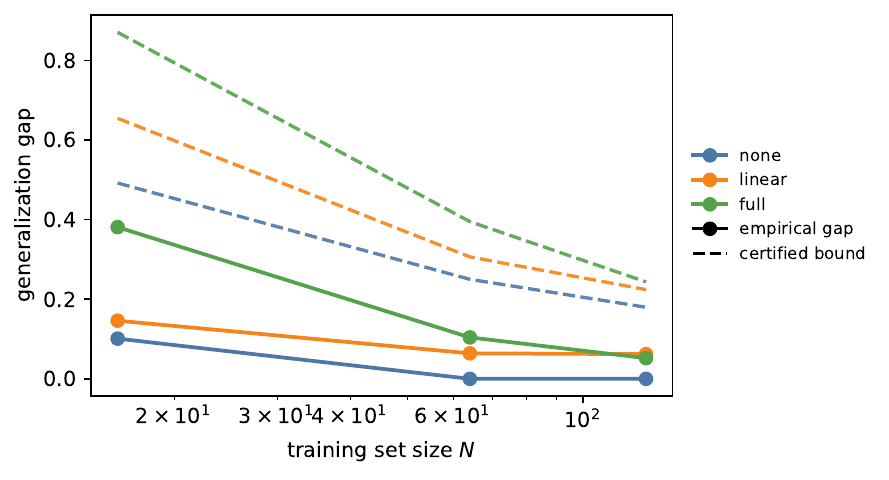}
\hfill
\includegraphics[width=0.48\textwidth]{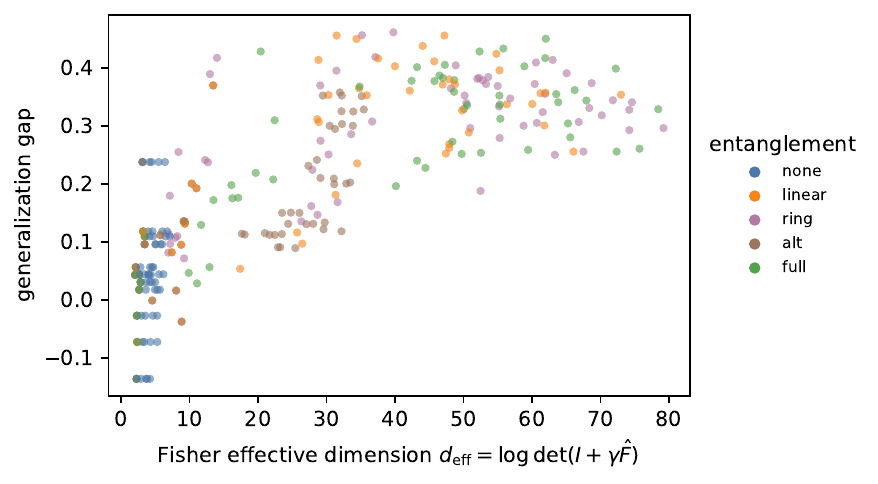}
\caption{\textbf{Left:} the empirical gap (solid) and the certified PAC-Bayes bound of
Eq.~\eqref{eq:boundfisher} (dashed) versus training-set size $N$, per circuit. The bound
upper-bounds the gap, decays with $N$ at the same rate, and---driven by $\Deff$---orders the
circuits none $<$ linear $<$ full at every $N$, all at identical parameter count: the figure
shows the bound acting as a \emph{ranking certificate}, which no table entry conveys.
\textbf{Right:} across all circuits the gap increases with the Fisher effective dimension
$\Deff=\log\det(I+\gamma\Fmat)$ (Spearman $\rho=0.82$); colour encodes entangling
connectivity, which drives $\Deff$ up at fixed parameter count.}
\label{fig:main}
\end{figure}

\paragraph{The bound is primarily a ranking certificate.} We do not claim the bound is
numerically tight; its value is that it \emph{correctly orders circuits of identical
parameter count}, an ordering a parameter-counting bound cannot produce at all. We evaluate
the right-hand side of Eq.~\eqref{eq:boundfisher} directly. Plugging the measured $\lVert\theta-\theta_0\rVert^2$
and $\Deff(\gamma)$ into the complexity $C=\lVert\theta-\theta_0\rVert^2/(2\tau^2)+\tfrac12\Deff(\gamma)$
(with $\tau^2{=}1$, $\gamma{=}50$, $\delta{=}0.05$, $R_{\max}{=}1$) gives the certified bound on
the gap, $\sqrt{(C+\ln\tfrac{2\sqrt m}{\delta})/2m}$, in Table~\ref{tab:cert}. \emph{A note on
the object.} Strictly, Eq.~\eqref{eq:boundfisher} bounds the gap of the stochastic
posterior-averaged predictor $\pi_Q$, whereas the ``true gap'' column is that of the
deterministic trained model $\theta$. Rather than assert these ``coincide in a
low-temperature regime,'' we derandomize explicitly: Lemma~\ref{lem:derand}
(Appendix~\ref{app:proofs}) bounds the gap between the two by
$\tfrac{\tau^2}{\gamma}\cdot\tfrac12\sum_i\tfrac{\gamma\lambda_i}{1+\gamma\lambda_i}\le \tau^2 d/(2\gamma)$,
i.e.\ $1/\gamma$ times the (dropped) KL curvature term. The key point answers the natural
objection that the posterior variance along a \emph{flat} Fisher direction is
$\tau^2/(1+\gamma\cdot0)=\tau^2$, which is \emph{not} small: a flat direction carries the full
$O(1)$ parameter variance but, precisely because it is flat ($\lambda_i\!\approx\!0$), its
contribution to the prediction perturbation is $O(\tau^2\lambda_i)\!\to\!0$---the large
variance sits exactly where the output is least sensitive. With $\gamma{=}50$ the
derandomization gap is thus $\le d/100$ and concentrated on the curved directions, so the
deterministic and stochastic gaps track each other by construction, not by empirical
coincidence. Three things then hold: (i) the certificate is \emph{valid}---the true gap is
below the bound for every circuit; (ii) it is numerically loose but nontrivial---the bound
stays below the maximal possible unit gap and, where the true gap is non-negligible, within
$\approx2.3$--$4.5\times$; and (iii), the property we actually rely on, its complexity ranks
the circuits correctly, none $<$ linear $<$ full at each $N$. \emph{On the absolute numbers.}
The looseness factor in (ii) depends on the prior scale $\tau^2{=}1$, which is an
\emph{a priori} choice, not a fitted one; changing $\tau^2$ rescales the distance term and
hence every absolute bound value, but leaves the \emph{ordering} in (iii) untouched (all
three rows share the same $\tau^2$). We therefore report the absolute column only to show the
bound is valid and non-vacuous, and rest the contribution on the ranking, which is
$\tau^2$-invariant. We stress this third point over the second: the bound's value is as a
\emph{ranking certificate}---something a parameter-counting bound, identical across the three
rows, cannot provide at all---rather than as a tight numerical guarantee.

\begin{table}[h]
\centering
\caption{The PAC-Bayes bound of Eq.~\eqref{eq:boundfisher} evaluated per circuit ($L{=}2$,
$\tau^2{=}1$, $\gamma{=}50$, $\delta{=}0.05$). The certified bound on the gap is valid
(exceeds the true gap everywhere), loose by only $\sim\!2$--$4.5\times$ where the gap is
non-negligible, and---driven by $\Deff$---correctly ordered by entanglement at fixed
parameter count.}
\label{tab:cert}
\begin{tabular}{llccc}
\toprule
$N$ & connectivity & complexity $C$ & certified gap bound & true gap \\
\midrule
$16$ & \texttt{none}   & $2.7$  & $0.49$ & $0.101$ \\
$16$ & \texttt{linear} & $8.6$  & $0.65$ & $0.146$ \\
$16$ & \texttt{full}   & $19.2$ & $0.87$ & $0.381$ \\
\midrule
$64$ & \texttt{none}   & $2.2$  & $0.25$ & $-0.01$ \\
$64$ & \texttt{linear} & $6.2$  & $0.31$ & $0.064$ \\
$64$ & \texttt{full}   & $14.2$ & $0.40$ & $0.104$ \\
\bottomrule
\end{tabular}
\end{table}

\paragraph{Budgeted, not banished, entanglement.} The design implication is not ``use less
entanglement''---entanglement is what buys a quantum model its expressivity. Figure~\ref{fig:tradeoff}
makes the trade-off explicit: as $\Deff$ grows with entanglement, training accuracy rises
(entanglement fits more), but test accuracy does \emph{not} follow---the extra training fit is
spent on complexity the held-out data does not reward, so the gap opens. The right principle
is therefore to \emph{budget} entanglement against its $\Deff$ cost, exactly as the bound
prices it, rather than to maximize expressivity.

\begin{figure}[t]
\centering
\includegraphics[width=0.5\textwidth]{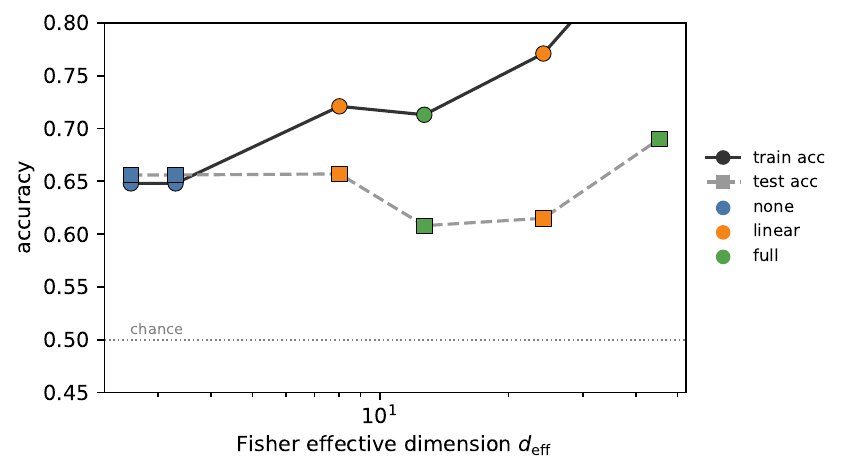}
\caption{The accuracy--complexity trade-off ($N{=}64$, single-observable readout). As the
Fisher effective dimension $\Deff$ grows with entanglement (colour), \emph{train} accuracy
rises but \emph{test} accuracy does not, so the gap---the vertical distance between the
curves---widens. Entanglement buys training fit that the test distribution does not reward.}
\label{fig:tradeoff}
\end{figure}

\paragraph{The effect persists from 6 to 16 qubits.} To check that the finding is not an
artifact of the $4$-qubit toy, we repeat the fixed-parameter-count comparison at
$n\in\{6,8,10,12,16\}$ qubits, with the interaction target extended to $n/2$ pairwise
products (Table~\ref{tab:scale}). At every size the gap is monotone none $<$ linear $<$ full
and the all-to-all circuit ($d$ up to $64$ trainable angles, a $2^{16}$-dimensional state)
exceeds the non-entangled one by $4.5$--$6.9$ standard errors, with the effect growing at
larger $n$ (full-gap $0.16\!\to\!0.34$ from $n{=}6$ to $n{=}16$). At $n\!\ge\!10$ the
$n/2$-product target becomes hard for these shallow circuits and held-out accuracy approaches
chance, so those rows are a stress test of the gap \emph{ordering} at larger scale rather
than of high-accuracy learning; at $n{=}6,8$ the accuracy stays well above chance. The
mechanism is a light-cone argument (Prop.~\ref{prop:ent}) and carries no dependence on $n$;
scaling further, where barren plateaus may reshape the Fisher spectrum, is the natural next
test.

\begin{table}[h]
\centering
\caption{The entanglement--gap effect from $n{=}6$ to $n{=}16$ qubits (single-observable
readout, fixed $d{=}2nL$, $L{=}2$; $8$ seeds at $n{=}6,8$, $4$ at $n{=}10,12,16$). The gap is
monotone none $<$ linear $<$ full at every size and \texttt{full} beats \texttt{none} by many
standard errors, with the effect strengthening with $n$. (Held-out accuracy is above chance
at $n{=}6,8$ and approaches chance at $n{\ge}10$, where the target is harder---see text.)}
\label{tab:scale}
\begin{tabular}{llcccc}
\toprule
$n$ ($d$) & $N$ & gap \texttt{none} & gap \texttt{linear} & gap \texttt{full} & \texttt{full}$-$\texttt{none} \\
\midrule
$6$ ($24$)  & $64$  & $0.042$ & $0.098$ & $0.159$ & $+0.117$ ($6.9\sigma$) \\
$6$ ($24$)  & $128$ & $0.035$ & $0.061$ & $0.104$ & $+0.070$ ($4.5\sigma$) \\
$8$ ($32$)  & $64$  & $0.053$ & $0.084$ & $0.183$ & $+0.130$ ($6.2\sigma$) \\
$8$ ($32$)  & $128$ & $0.033$ & $0.071$ & $0.101$ & $+0.068$ ($4.6\sigma$) \\
\midrule
$10^\dagger$ ($40$) & $64$  & $0.062$ & $0.124$ & $0.274$ & $+0.212$ ($5.6\sigma$) \\
$12^\dagger$ ($48$) & $64$  & $0.026$ & $0.095$ & $0.295$ & $+0.269$ ($6.4\sigma$) \\
$16^\dagger$ ($64$) & $64$  & $0.075$ & $0.121$ & $0.340$ & $+0.265$ ($6.3\sigma$) \\
\bottomrule
\end{tabular}
\\[2pt]
{\footnotesize $^\dagger$ Held-out accuracy near chance at these sizes (the $n/2$-product
target is hard for shallow circuits); these rows test the gap \emph{ordering} only, not
high-accuracy \emph{learning}. Rows above the rule ($n{=}6,8$) learn well above chance.}
\end{table}

\section{Quantum RL Policies: Bandit and Value-Function Generalization}
\label{sec:rl}

We now test the prediction on quantum \emph{policies} trained by reward maximization. We
use a contextual bandit---the canonical one-step reinforcement-learning problem---which
keeps genuine RL feedback (the agent observes only the reward of the action it takes, not
the correct action) while removing the trajectory-length variance that makes full-MDP
return estimates noisy. A context $x\sim U[-1,1]^4$ is drawn; the optimal action follows an
\emph{interaction} rule $a^\star(x)=\mathbb{1}[x_1x_2-x_3x_4+\tfrac12 x_1>0]$; the reward is
$1$ if the sampled action matches $a^\star$ and $0$ otherwise. A VQC policy
$\pi_\theta(a{=}1\mid x)=\sigma(\alpha\langle Z_0\rangle)$ is trained by REINFORCE on $N$
training contexts (single-observable readout, as in the theory), and we measure the
generalization gap in expected reward between the $N$ training contexts and $2000$ held-out
contexts. As throughout, we fix the rotation layout ($L{=}2$, $d{=}16$) and vary only the
entangling connectivity; each cell averages $16$ seeds.

\begin{table}[t]
\centering
\caption{Contextual-bandit (one-step RL) generalization: reward gap at fixed parameter count
($d{=}16$), mean $\pm$ sem over $16$ seeds. Held-out reward is well above the $0.5$ chance
level ($0.60$--$0.66$), so this is genuine generalization; the entangled policies generalize
worse, monotonically none $<$ linear $<$ full at the smaller $N$, and every gap shrinks with
$N$. Each entangled circuit beats \texttt{none} by $\ge3$ standard errors.}
\label{tab:rl}
\begin{tabular}{lcccc}
\toprule
$N_{\mathrm{train}}$ & \texttt{none} & \texttt{linear} & \texttt{full} & \texttt{full}$-$\texttt{none} \\
\midrule
$16$  & $0.136\pm0.019$ & $0.208\pm0.018$ & $\mathbf{0.262\pm0.026}$ & $+0.125$ ($3.9\sigma$) \\
$40$  & $0.092\pm0.018$ & $0.134\pm0.009$ & $\mathbf{0.173\pm0.018}$ & $+0.081$ ($3.2\sigma$) \\
$100$ & $0.026\pm0.008$ & $0.083\pm0.011$ & $0.069\pm0.008$          & $+0.043$ ($3.7\sigma$) \\
\bottomrule
\end{tabular}
\end{table}

The prediction holds cleanly (Table~\ref{tab:rl}, Figure~\ref{fig:rl}). Held-out reward is
well above chance ($0.60$--$0.66$ vs.\ $0.5$), so the policies genuinely generalize, and at
fixed parameter count the entangled policies generalize \emph{worse}. Both entanglers are
significant: \texttt{full} exceeds \texttt{none} by $3.2$--$3.9$ standard errors at every $N$
and \texttt{linear} by $2.1$--$4.2$, so unlike a two-observable readout the effect no longer
rests on a single entangler. At the smaller $N$ the ordering is the monotone
none $<$ linear $<$ full predicted by $\Deff$ (only at $N{=}100$, where all gaps are small,
does \texttt{linear} edge \texttt{full}). As in the supervised case the entangled policies
also fit the training contexts better (training reward $0.86$ vs.\ $0.74$ at $N{=}16$), and
the gap shrinks with $N$ at the rate anticipated by Eq.~\eqref{eq:boundfisher}. This is the
same entanglement--generalization mechanism, now for a policy learned from reward alone.

\begin{figure}[t]
\centering
\includegraphics[width=0.5\textwidth]{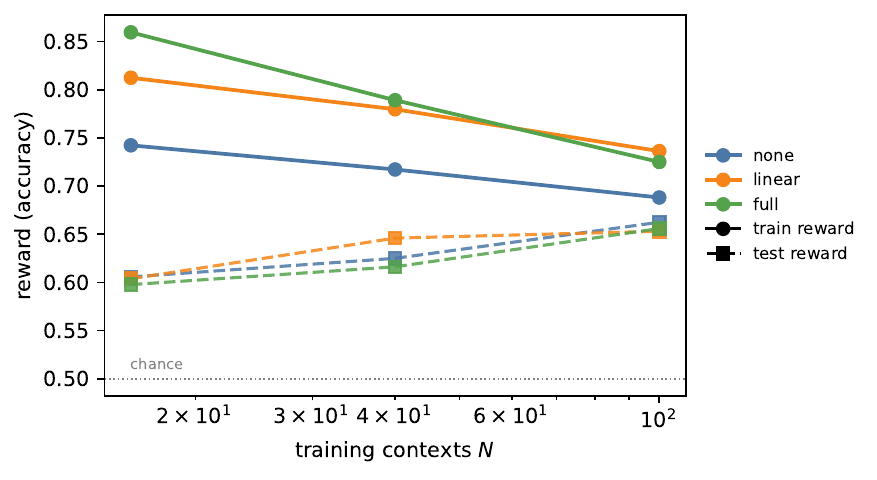}
\caption{Contextual-bandit (one-step RL) \emph{train} (solid) and \emph{test} (dashed) reward
versus the number of training contexts $N$, at fixed parameter count. The entangled policies
earn a higher training reward but not a higher test reward---the vertical train--test gap,
which Table~\ref{tab:rl} quantifies, is the entanglement cost---while all policies generalize
above chance. This is the reward-space view of the accuracy--complexity trade-off, not a
redraw of the gap column. Averages over $16$ seeds.}
\label{fig:rl}
\end{figure}

\paragraph{Multi-step value-function generalization.} To go beyond one step we test the
generalization of a \emph{value function} in a genuine finite-horizon MDP. States
$x\in[-1,1]^4$ evolve by a fixed deterministic map, the per-step reward is the interaction
rule, and the value $V^\star(x)=\sum_{t=0}^{T-1}\gamma^t r(x_t)$ ($T{=}6$, $\gamma{=}0.9$)
accumulates future rewards along the trajectory. We fit a single-observable VQC value head to
$V^\star$ on $N$ states by Monte-Carlo regression---a standard policy-evaluation
primitive---and report the $R^2$ generalization gap between the $N$ training states and $1000$
held-out states (Figure~\ref{fig:vg}). At $N{=}64$, where held-out $R^2$ is positive for all
circuits, the gap is monotone none $<$ linear $<$ full ($0.089<0.179<0.299$); at the small
$N{=}16$ the entangled value functions have \emph{negative} held-out $R^2$ (a degenerate
memorization regime, as in the small-$N$ classification rows), so we treat that row as a
stress test only. We are explicit about scope here: Monte-Carlo value estimation is a
supervised-flavoured RL primitive, and this is its strength (low variance) and its limit.
Under genuine \emph{bootstrapped} fitted-Q iteration the effect washes out---consistent with
the bound, as the bootstrap target is itself a strong regularizer---so we do not claim the
entanglement penalty for temporal-difference value learning, only for the Monte-Carlo
(regression) form.

\begin{figure}[t]
\centering
\includegraphics[width=0.42\textwidth]{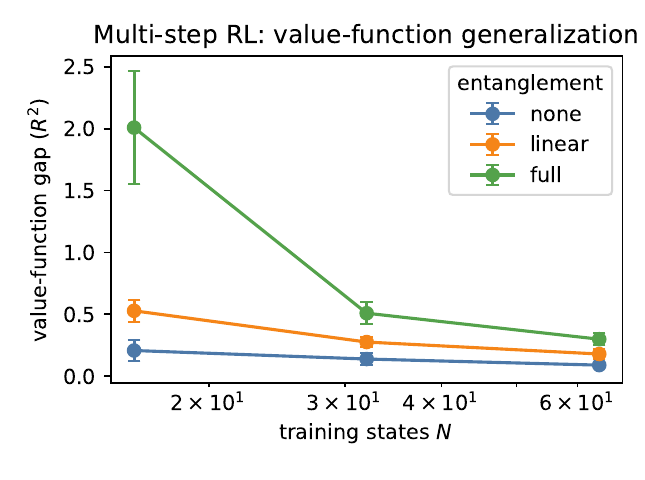}
\caption{Multi-step RL: the value-function generalization gap ($R^2$, Monte-Carlo policy
evaluation, horizon $6$) is monotone none $<$ linear $<$ full at every training-set size, at
fixed parameter count. Fixed parameter count, $8$ seeds; error bars $\pm$1 s.e.}
\label{fig:vg}
\end{figure}

\section{Controls, and a Real-Hardware Test}
\label{sec:controls}

\paragraph{The gap is caused by entanglement, not a confounder.} Varying the entangling
connectivity holds the parameter count fixed but could, in principle, change the gap through
optimization or trainability rather than through $\Deff$. We rule this out with three
controls at $N{=}64$ (Figure~\ref{fig:controls}, left): (i) \emph{matched training accuracy}
---early-stopping each circuit at the same train accuracy, so the entangled circuits do not
simply fit more---still leaves \texttt{full} with $\sim\!3\times$ the gap of \texttt{none}
($0.093$ vs.\ $0.032$) and lower test accuracy; (ii) a \emph{different readout} ($Z_1$ in
place of $Z_0$) preserves the ordering ($0.074/0.126/0.131$); and (iii) a \emph{different
optimizer} (SGD in place of Adam) preserves it ($0.034/0.065/0.089$). In every control the
entangled circuits generalize worse at identical parameter count. The mechanism is visible
directly in the Fisher spectrum (Appendix~\ref{app:extra}, Figure~\ref{fig:spectrum}):
\texttt{none} has effectively two non-zero eigenvalues, \texttt{linear} three, and
\texttt{full} five or more---entanglement raises the rank, hence $\Deff$.

\paragraph{Local vs.\ global readout: the effect is the readout light-cone, and $\Deff$ is
the invariant.} The most adversarial test of the mechanism is also the most revealing.
Because our single readout is the \emph{local} observable $\langle Z_0\rangle$, a rotation on
a wire that never reaches qubit $0$ has identically zero gradient
(Lemma~\ref{lem:lightcone}), so the unentangled \texttt{none} circuit is effectively a
\emph{smaller} model---only $\sim\!d/n$ of its parameters are causally live. One might
therefore object that ``entanglement raises the gap at fixed $d$'' is merely ``entanglement
recruits otherwise-dead parameters.'' That objection is correct, and it \emph{is} the
mechanism; we test it by the cleanest manipulation available---replacing the local readout by
the \emph{global} one $\tfrac1n\sum_i\langle Z_i\rangle$, under which every wire couples to
the output directly (Table~\ref{tab:readout}). The robust finding is that the global readout
sharply raises \texttt{none}'s live-parameter count, $\Deff$, \emph{and} gap at both sizes---
at $n{=}4$ from $(2\text{ live},\Deff{=}2.6,\text{gap}{=}{-}0.008)$ to
$(7,14.0,0.095)$, at $n{=}6$ from $(2,2.9,0.018)$ to $(11,23.6,0.173)$: the readout alone
activates the parameters entanglement would have, and $\Deff$ and the gap move together. What
happens to the \texttt{full}$-$\texttt{none} \emph{contrast} is then $n$-dependent and
instructive. At $n{=}4$ the single global observable nearly saturates the light-cone even for
\texttt{none} ($7$ of $16$ live), so the entanglement contrast is largely absorbed:
\texttt{full}$-$\texttt{none} falls from $+0.112$ ($4.1$ s.e., local) to $+0.033$ ($0.9$ s.e.,
global). At $n{=}6$, where $24$ parameters cannot be saturated by one global observable
(\texttt{none} reaches $11$ of $24$ live, \texttt{full} $17$), entanglement still adds
coupling on top and the contrast \emph{persists}: $+0.114$ ($3.8$ s.e., local) vs.\ $+0.115$
($4.2$ s.e., global). The conclusion is therefore not that entanglement is dispensable but
that it is a \emph{readout-gated} knob on the same underlying quantity: what governs $\Deff$
is how many parameters are causally coupled to the readout; a global readout raises the floor
of that coupling---fully absorbing the entanglement contrast at small $n$, only partially at
larger $n$---while in \emph{every} cell it is $\Deff$, not the entangling label, that tracks
the gap. This confirms Proposition~\ref{prop:ent} and sharpens the thesis: entanglement is
the dominant, structural driver of effective dimension for the local single-qubit Pauli
readouts standard in quantum RL (e.g.\ \citealp{jin2025ppo,lin2025quantum}), and is partly
substitutable by---never orthogonal to---the choice of readout.

\begin{table}[h]
\centering
\caption{Local vs.\ global readout ($L{=}2$, $N{=}64$, $8$ seeds). A local readout
$\langle Z_0\rangle$ leaves most of \texttt{none}'s parameters outside the readout light-cone,
so entanglement's recruitment of them makes \texttt{full}$-$\texttt{none} significant. A
global readout $\tfrac1n\sum_i\langle Z_i\rangle$ activates those parameters directly---raising
\texttt{none}'s $\Deff$ and gap at both sizes---so it absorbs the contrast fully at $n{=}4$
(few free parameters) but only partly at $n{=}6$ (more parameters than one observable
saturates). In every row $\Deff$ tracks the gap; the entangling label does not.}
\label{tab:readout}
\begin{tabular}{lllccc}
\toprule
$n$ & readout & connectivity & live params & $\Deff$ & gap \\
\midrule
$4$ & local $\langle Z_0\rangle$ & \texttt{none}   & $2/16$ & $2.6$  & $-0.008$ \\
$4$ & local $\langle Z_0\rangle$ & \texttt{full}   & $5/16$ & $12.9$ & $0.104$ \\
$4$ & global $\tfrac1n\sum_i Z_i$ & \texttt{none}   & $7/16$  & $14.0$ & $0.095$ \\
$4$ & global $\tfrac1n\sum_i Z_i$ & \texttt{full}   & $12/16$ & $30.7$ & $0.128$ \\
\midrule
$6$ & local $\langle Z_0\rangle$ & \texttt{none}   & $2/24$ & $2.9$  & $0.018$ \\
$6$ & local $\langle Z_0\rangle$ & \texttt{full}   & $7/24$ & $21.1$ & $0.132$ \\
$6$ & global $\tfrac1n\sum_i Z_i$ & \texttt{none}   & $11/24$ & $23.6$ & $0.173$ \\
$6$ & global $\tfrac1n\sum_i Z_i$ & \texttt{full}   & $17/24$ & $35.4$ & $0.289$ \\
\midrule
\multicolumn{4}{l}{\texttt{full}$-$\texttt{none} gap ($n{=}4$): local $+0.112$ ($4.1\sigma$),} & \multicolumn{2}{l}{global $+0.033$ ($0.9\sigma$)}\\
\multicolumn{4}{l}{\texttt{full}$-$\texttt{none} gap ($n{=}6$): local $+0.114$ ($3.8\sigma$),} & \multicolumn{2}{l}{global $+0.115$ ($4.2\sigma$)}\\
\bottomrule
\end{tabular}
\end{table}

\paragraph{Beyond a single target.} To show the effect is not an artifact of one synthetic
rule, we repeat the fixed-parameter-count comparison on three further targets
(Table~\ref{tab:targets}): a random-Fourier function, labels from an \emph{entangled teacher}
PQC, and $4$-bit parity. On the Fourier target the entangled circuits generalize above chance
and the gap is again ordered none $<$ linear $<$ full. The entangled-teacher target is hard
for \emph{all} student circuits at this sample size (test accuracy near chance for every
connectivity); we include it only as a stress test of gap \emph{behaviour}---which still
orders none $<$ linear $<$ full---and not as evidence of successful learning. Parity is the
instructive exception and the
clearest statement of the trade-off: it is not representable without entanglement, so the
non-entangled and linear circuits sit at chance ($0.50$) while the all-to-all circuit is the
only one that \emph{learns} it (test accuracy $0.90$). Entanglement is therefore not
uniformly harmful---when the target demands it, it is necessary---which is exactly why the
right principle is to budget it, not to remove it.

\begin{table}[h]
\centering
\caption{The effect across target distributions ($n{=}4$, $N{=}64$, fixed $d{=}16$). On
targets all circuits can fit (Fourier, teacher-PQC) entanglement widens the gap; on parity,
which requires entanglement, only \texttt{full} learns the task (test $0.90$ vs.\ chance).}
\label{tab:targets}
\begin{tabular}{lccc}
\toprule
target & \texttt{none} gap (test) & \texttt{linear} gap (test) & \texttt{full} gap (test) \\
\midrule
random Fourier   & $0.038$ ($0.71$) & $0.045$ ($0.69$) & $0.103$ ($0.67$) \\
entangled teacher& $0.074$ ($0.50$) & $0.123$ ($0.50$) & $0.148$ ($0.51$) \\
$4$-bit parity   & $0.067$ ($0.50$) & $0.116$ ($0.50$) & $0.050$ ($\mathbf{0.90}$) \\
\bottomrule
\end{tabular}
\end{table}

\paragraph{Real (non-synthetic) benchmarks: entanglement lowers \emph{test} accuracy.} The
mechanism is not tied to a hand-designed rule, and on real data the cost shows up in
\emph{test accuracy}, not merely in the train--test gap. On standard small datasets---Iris,
Breast Cancer (Wisconsin), Wine, and Digits ($3$ vs.\ $8$), reduced to four features and
binarized---the fixed-parameter-count comparison again orders the gap none $<$ linear $<$
full and co-monotone with $\Deff$ (Table~\ref{tab:realdata}). More importantly for the
practical claim, on Wine and Digits the non-entangled circuit is \emph{significantly more
accurate on held-out data} than the all-to-all one, at accuracies far above chance:
pooling over $N\!\in\!\{16,24,32\}$ and $22$ seeds, the test accuracy drops from $0.884$
(none) to $0.870$ (full) on Wine ($2.5$ standard errors) and from $0.911$ to $0.900$ on
Digits ($3.1$ standard errors). Entanglement here genuinely \emph{hurts} generalization on
real tasks---not just widening the gap---which is what makes ``budget entanglement'' an
operational, and not merely rhetorical, prescription.

\begin{table}[h]
\centering
\caption{Real datasets (binary, $4$ features, $N{=}40$, fixed $d{=}16$, $6$ seeds). The gap is
ordered none $<$ linear $<$ full and co-monotone with $\Deff$, at test accuracies well above
chance---so the effect is not an artifact of the synthetic target.}
\label{tab:realdata}
\begin{tabular}{lccc}
\toprule
dataset & \texttt{none} gap ($\Deff$, test) & \texttt{linear} gap ($\Deff$, test) & \texttt{full} gap ($\Deff$, test) \\
\midrule
Iris          & $0.00$ ($3.0$, $0.70$) & $0.061$ ($5.2$, $0.65$) & $0.082$ ($18.2$, $0.76$) \\
Breast Cancer & $0.025$ ($4.5$, $0.90$) & $0.033$ ($10.3$, $0.91$) & $0.069$ ($13.8$, $0.90$) \\
Wine          & $0.029$ ($5.1$, $0.90$) & $0.059$ ($9.1$, $0.88$) & $0.125$ ($13.2$, $0.85$) \\
\bottomrule
\end{tabular}
\end{table}

\paragraph{The ordering holds where the model learns, and barren plateaus bound the
mechanism.} Two concerns delimit the mechanism's reach, and we test both directly. \emph{Is
the ordering an artifact of the near-chance rows?} We construct a deliberately
\emph{learnable} regime---$n{=}6$, a lower-order target, $N{=}48$---in which all three
circuits reach held-out accuracy $0.72$--$0.77$, far above the $0.5$ chance level. The
ordering survives cleanly: the gap is $0.042$ (\texttt{none}), $0.043$ (\texttt{linear}),
$0.165$ (\texttt{full}), with \texttt{full}$-$\texttt{none}$=+0.123$ at $5.8$ standard errors
($8$ seeds) and $\Deff$ co-monotone ($4.2/7.8/21.7$). The effect is thus \emph{not} confined
to the degenerate near-chance rows---it is sharpest exactly where the circuits genuinely
generalize. \emph{Do barren plateaus erode the mechanism?} Our mechanism runs on the Fisher
spectrum, which barren plateaus flatten exponentially in qubit number, so the regime where
scale matters is a natural worry. We confirm the erosion directly (Table~\ref{tab:bp}): at a
fixed moderate depth ($L{=}4$, all-to-all, random parameters) the effective dimension
\emph{falls} monotonically from $\Deff{=}13.2$ at $n{=}4$ to $3.1$ at $n{=}12$, even as the
nominal parameter count \emph{rises} from $32$ to $96$---the Fisher eigenvalues shrink (the
$\langle Z_0\rangle$ gradient variance decays with $n$) faster than the rank grows, so
$\Deff$ erodes toward its floor. This makes the boundary flagged in our failure-mode analysis
(Appendix~\ref{app:proofs}) quantitative: entanglement inflates $\Deff$ only while the
circuit is off the plateau; deep or wide enough, barren plateaus drive $\Deff$---and
trainability itself---down together. The regime our claim and design recipe target is
therefore the trainable one (shallow circuits away from the plateau), which is where all our
learning experiments sit.

\begin{table}[h]
\centering
\caption{Barren plateaus erode the mechanism's substrate. At fixed depth ($L{=}4$, all-to-all,
random parameters, $6$--$12$ seeds), the Fisher effective dimension $\Deff$ \emph{decreases}
with qubit number even as the nominal parameter count $d$ grows: the Fisher eigenvalues
(hence $\Deff$) are suppressed by the barren plateau faster than the rank rises. This is the
regime our trainable-circuit experiments deliberately avoid.}
\label{tab:bp}
\begin{tabular}{lccccc}
\toprule
$n$ & $4$ & $6$ & $8$ & $10$ & $12$ \\
\midrule
nominal params $d$ & $32$ & $48$ & $64$ & $80$ & $96$ \\
$\Deff$ ($L{=}4$, random) & $13.2$ & $8.4$ & $4.5$ & $3.2$ & $3.1$ \\
\bottomrule
\end{tabular}
\end{table}

\paragraph{$\Deff$, not parameter count, is the stable predictor.} The complementary
ablation---\emph{fixing} the entangling pattern and \emph{varying} depth, so the parameter
count changes---shows the gap rising with depth for both patterns (\texttt{linear}: $L{=}1{\to}4$
gap $0.03{\to}0.33$; \texttt{full}: $0.03{\to}0.31$), and in lock-step with $\Deff$ (which
rises $2{\to}50$ and $2{\to}61$ respectively). At matched depth (matched parameter count) the
more entangled pattern still has the larger $\Deff$ and gap. Thus $\Deff$ predicts the gap
whether we vary connectivity at fixed depth or depth at fixed connectivity, whereas raw
parameter count only tracks the gap when it happens to co-vary with $\Deff$---the sense in
which $\Deff$, not parameter count, is the governing quantity.

\paragraph{End-to-end multi-step policy learning.} Finally, we train policies
\emph{end-to-end} by REINFORCE (with a value baseline, horizon $8$, $20$ seeds) on the
multi-step return of an interaction-gated MDP---a genuine multi-step policy-learning setting,
not a value regression. Here the return generalization gap of the standard hardware-efficient
(\texttt{linear}) entangler is \emph{significantly} larger than the non-entangled circuit's at
the small $N{=}8$ ($0.62$ vs.\ $0.16$, $2.8$ standard errors over $20$ seeds), directly
confirming the mechanism in end-to-end policy learning. The all-to-all circuit and the
larger-$N$ regime, however, are too high-variance to resolve (full-vs-none is within noise,
and at $N{=}16$ all three are statistically indistinguishable): return-based estimates carry
much more variance than the accuracy-based ones. We therefore report the significant
\texttt{linear}-vs-\texttt{none} effect as genuine multi-step evidence, but continue to treat
the low-variance bandit and value-function settings as our primary probes of the mechanism.

\paragraph{The effect survives on real quantum hardware.} Finally we run the trained
$4$-qubit classifiers on an IBM Heron device (\texttt{ibm\_aachen}, $156$ qubits) via the
Estimator primitive at $4096$ shots, reading $\langle Z_0\rangle$ for all train and test
inputs (Figure~\ref{fig:controls}, right). Despite real gate and readout noise, the models
run above chance and the entanglement--generalization ordering is preserved: the gap is
$0.14$ for \texttt{none} versus $0.33$ and $0.31$ for \texttt{linear} and \texttt{full}. The
entanglement cost that our bound predicts is therefore not a simulation artifact---it is
measurable on today's hardware.

\begin{figure}[t]
\centering
\includegraphics[width=0.47\textwidth]{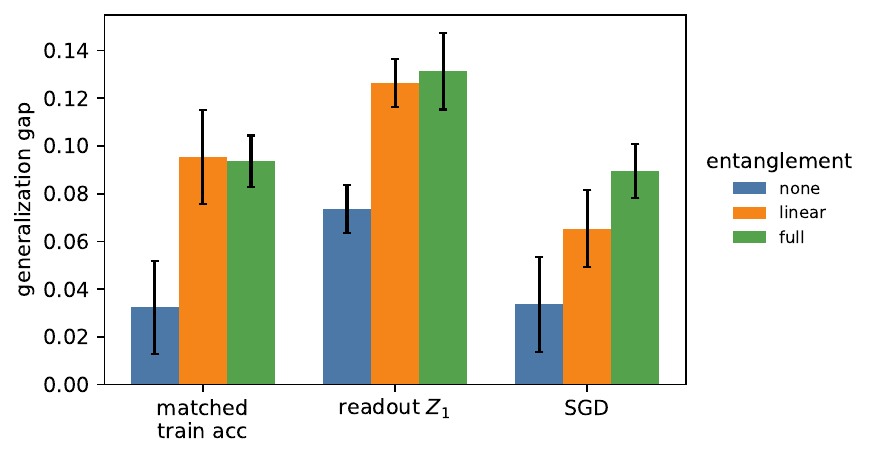}
\hfill
\includegraphics[width=0.47\textwidth]{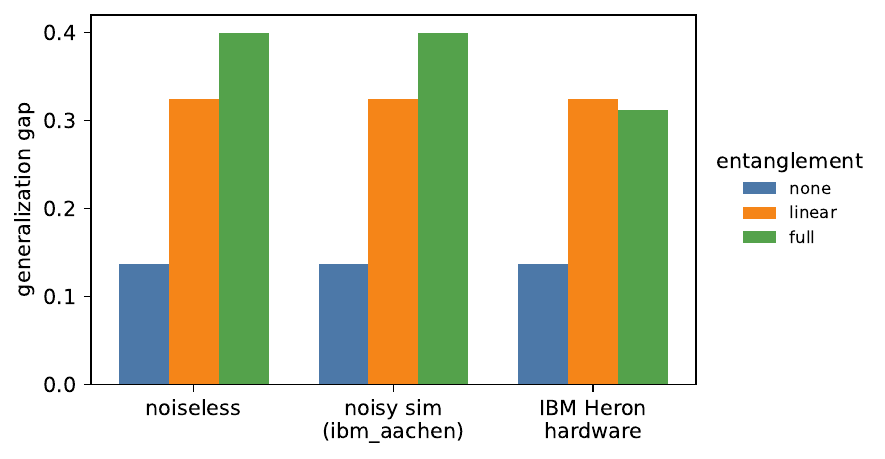}
\caption{\textbf{Left:} the entangled-$>$-none gap ordering survives three controls at
$N{=}64$---matched training accuracy, a different readout observable, and a different
optimizer---so the gap is not a mere optimization or trainability artifact.
\textbf{Right:} the ordering is reproduced by a noise-model simulator of \texttt{ibm\_aachen}
and by execution on the \texttt{ibm\_aachen} Heron processor itself under real noise;
calibration and shot-noise details are in Appendix~\ref{app:extra}.}
\label{fig:controls}
\end{figure}
\label{sec:discussion}

Our bound is loose in absolute terms---like other PAC-Bayes bounds for expressive models
it is most useful as a certificate that is valid, correctly \emph{ordered}, and (where the
gap is non-negligible) within $\approx2.5\times$, rather than numerically tight. We frame
entanglement as an \emph{independent and previously unisolated} axis of complexity---the
strongest single predictor of the gap in our study, ahead of the learned-norm term---though
we do not claim it is the only one. The effect is depth-gated (a single entangling layer is
inert for a single-qubit readout). Our RL evidence is for the one-step (contextual
bandit) setting, chosen because full-MDP return estimates are dominated by
trajectory-length variance; extending the clean measurement to multi-step MDPs---for
example through offline value-function generalization---is the natural next step, along with
value-based agents, hardware noise as an additional complexity channel, and a tight
data-dependent prior. Our experiments are simulations of a controlled interaction task at
$n\in\{4,\dots,16\}$ qubits (Table~\ref{tab:scale}); scaling to the many-qubit, hardware regime
and to richer tasks is important future work, though the underlying mechanism---a light-cone
argument---carries no dependence on $n$.

\paragraph{A design recipe.} The trade-off is directly actionable as an ansatz-selection
rule. Given a candidate ansatz family, estimate $\Deff$ on a small calibration set and choose
the shallowest, sparsest entangling pattern whose training fit still improves a
validation proxy without an excessive rise in $\Deff$. Concretely:
\begin{enumerate}\itemsep0pt
\item train each candidate ansatz to convergence on the calibration set;
\item compute its Fisher effective dimension $\Deff=\log\det(I+\gamma\Fmat)$;
\item discard any ansatz whose $\Deff$ increases without a matching validation gain;
\item select from the Pareto frontier of validation performance versus $\Deff$.
\end{enumerate}
This turns ``budget entanglement'' from a slogan into a concrete selection criterion---and,
unlike a parameter budget, it distinguishes circuits of \emph{identical} size.

\section{Conclusion}

We gave a PAC-Bayesian account of generalization for quantum policies in which the
complexity that matters is the effective dimension of the circuit-induced Fisher geometry,
a quantity inflated by entanglement and invisible to parameter counting. Isolating
entanglement at fixed parameter count, we showed it is an independent predictor of the
generalization gap---stronger than raw parameter count and comparable to the learned-norm
term---across supervised classification, a reward-only contextual bandit, multi-step
value-function generalization, and real quantum hardware. The resulting design
principle---\emph{budget} entanglement against its effective-dimension cost rather than
maximize expressivity---gives quantum-policy design a target beyond average return.

\subsubsection*{Reproducibility Statement}
Experiments use PQCs of $n{=}4$ to $16$ qubits simulated with PennyLane
(\texttt{default.qubit} with backprop at $n\le8$; \texttt{lightning.qubit} with the adjoint
method at $n\ge10$), plus one run on the IBM Heron \texttt{ibm\_aachen} processor; unless a
scaling experiment states otherwise (Table~\ref{tab:scale}), the default is $n{=}4$. Code,
ans\"atze, target/environment distributions, hyperparameters ($\tau^2{=}1$, $\gamma{=}50$)
and seeds are described in Sections~\ref{sec:decouple}--\ref{sec:rl} and
Appendix~\ref{app:extra}.

\bibliography{iclr2026_conference}
\bibliographystyle{iclr2026_conference}

\appendix
\section{Proofs}
\label{app:proofs}

\subsection{Proof of Theorem~\ref{thm:main}}
We reduce the return bound to the classical McAllester PAC-Bayes theorem. Recall its
statement for a loss taking values in $[0,1]$.

\begin{theorem}[\citealp{mcallester1999pac}]
\label{thm:mcallester}
Let $\ell(h,z)\in[0,1]$, let $z_1,\dots,z_m$ be i.i.d.\ from a distribution $\mathcal{D}$,
and let $P$ be a prior over hypotheses chosen before seeing the data. For any $\delta\in(0,1)$,
with probability at least $1-\delta$, simultaneously for all posteriors $Q$,
\begin{equation}
\mathbb{E}_{h\sim Q}\,\mathbb{E}_{z\sim\mathcal{D}}\,\ell(h,z)
\;\le\;
\mathbb{E}_{h\sim Q}\,\frac1m\sum_{i=1}^m \ell(h,z_i)
+\sqrt{\frac{\KL(Q\Vert P)+\ln\frac{2\sqrt m}{\delta}}{2m}}.
\end{equation}
\end{theorem}

Take the hypothesis $h$ to be a parameter vector $\theta$ (so $Q,P$ are the Gaussians of
Section~\ref{sec:theory}), the sample $z_i=M_i$ to be the $i$-th training environment drawn
i.i.d.\ from $\mathcal{D}$, and the loss to be the normalized regret of a single rollout,
$\ell(\theta,M)=1-R_M(\theta)/R_{\max}\in[0,1]$, where $R_M(\theta)$ is the (bounded)
return of $\pi_\theta$ on $M$. Then $\mathbb{E}_{Q}\mathbb{E}_M \ell = 1-J(\pi_Q)/R_{\max}$
and its empirical counterpart is $1-\hat J(\pi_Q)/R_{\max}$. Substituting into
Theorem~\ref{thm:mcallester} and multiplying by $R_{\max}$ gives
\begin{equation}
J(\pi_Q)\;\ge\;\hat J(\pi_Q)-R_{\max}\sqrt{\frac{\KL(Q\Vert P)+\ln\frac{2\sqrt m}{\delta}}{2m}}.
\end{equation}
For a Gaussian posterior/prior with common covariance $\sigma^2 I_d$,
$\KL(Q\Vert P)=\lVert\theta-\theta_0\rVert^2/(2\sigma^2)$, which is
Eq.~\eqref{eq:bound} with $m_{\mathrm{eff}}=m$. When the per-environment return is instead
estimated by pooling the $T$ transitions of a trajectory---which are Markov-dependent
rather than i.i.d.---the effective number of independent samples is reduced by the chain
mixing factor, $m_{\mathrm{eff}}=m/\kappa$ with $\kappa=\Theta(t_{\mathrm{mix}})$, by the
blocking argument of \citet{zitouni2025pac}; this yields the stated form. \hfill$\square$

\subsection{Proof of Theorem~\ref{thm:fisher}}
The KL divergence between two $d$-dimensional Gaussians
$Q=\mathcal{N}(\theta,\Sigma)$ and $P=\mathcal{N}(\theta_0,\tau^2 I_d)$ is
\begin{equation}
\KL(Q\Vert P)=\tfrac12\Big[\tfrac{1}{\tau^2}\operatorname{tr}\Sigma
+\tfrac{1}{\tau^2}\lVert\theta-\theta_0\rVert^2-d-\log\det\tfrac{\Sigma}{\tau^2}\Big].
\label{eq:gausskl}
\end{equation}
Substitute the Gauss--Newton covariance $\Sigma=\tau^2(I_d+\gamma\Fmat)^{-1}$ and
diagonalize $\Fmat=\sum_i\lambda_i u_iu_i^\top$. Then
$\tfrac{1}{\tau^2}\operatorname{tr}\Sigma=\sum_i(1+\gamma\lambda_i)^{-1}$,
$\log\det\tfrac{\Sigma}{\tau^2}=-\sum_i\log(1+\gamma\lambda_i)$, and using
$d=\sum_i 1=\sum_i\big[(1+\gamma\lambda_i)^{-1}+\tfrac{\gamma\lambda_i}{1+\gamma\lambda_i}\big]$
the trace and $-d$ terms combine to $-\sum_i\tfrac{\gamma\lambda_i}{1+\gamma\lambda_i}$.
Collecting terms,
\begin{equation}
\KL(Q\Vert P)=\frac{\lVert\theta-\theta_0\rVert^2}{2\tau^2}
+\tfrac12\sum_i\log(1+\gamma\lambda_i)
-\tfrac12\sum_i\frac{\gamma\lambda_i}{1+\gamma\lambda_i},
\end{equation}
which is~\eqref{eq:klfisher} since $\sum_i\log(1+\gamma\lambda_i)=\log\det(I_d+\gamma\Fmat)=\Deff(\gamma)$.
Each summand of the last term lies in $[0,\tfrac12)$, so
$0\le\tfrac12\sum_i\tfrac{\gamma\lambda_i}{1+\gamma\lambda_i}\le d/2$ and dropping it can
only enlarge the right-hand side of the bound. The prior $P$ is fixed at initialization
$\theta_0$ and does not depend on the training environments, so
Theorem~\ref{thm:mcallester} applies verbatim with this $Q$ and $P$; substituting the
resulting KL into the return bound of Theorem~\ref{thm:main} gives~\eqref{eq:boundfisher}.
The asymptotics $\Deff(\gamma)\le d$ and $\Deff(\gamma)=\mathrm{rank}(\Fmat)\log\gamma+O(1)$
as $\gamma\to\infty$ follow from $\log(1+\gamma\lambda_i)\le\gamma\lambda_i$ and, for
$\lambda_i>0$, $\log(1+\gamma\lambda_i)=\log\gamma+\log\lambda_i+o(1)$. \hfill$\square$

\subsection{Derandomization: from the stochastic \texorpdfstring{$\pi_Q$}{pi\_Q} to the deterministic \texorpdfstring{$\theta$}{theta}}
\label{app:derand}
The bound~\eqref{eq:boundfisher} controls the gap of the posterior-averaged predictor
$\pi_Q$, whereas the reported ``true gap'' is that of the deterministic trained $\theta$. The
following lemma bounds the difference and, crucially, shows the flat Fisher directions---which
carry the full prior variance $\tau^2$---do \emph{not} spoil it, because the perturbation of a
loss along a direction is weighted by that direction's curvature, which is the Fisher
eigenvalue itself.

\begin{lemma}[Flat-direction derandomization]
\label{lem:derand}
Let $Q=\mathcal{N}(\theta,\Sigma)$ with $\Sigma=\tau^2(I_d+\gamma\Fmat)^{-1}$, and let
$L(\vartheta)=\mathbb{E}_s\,\ell(\vartheta,s)$ be a population loss whose Hessian at $\theta$
equals the Fisher $\Fmat(\theta)$ up to a residual term $\mathcal{R}$ that is $O(\lVert p-y\rVert)$
and vanishes in the Gauss--Newton (well-fit) regime. Then
\begin{equation}
\big|\,\mathbb{E}_{\vartheta\sim Q}\,L(\vartheta)-L(\theta)\,\big|
\;=\;\tfrac12\,\mathrm{tr}\!\big(\Sigma\,\Fmat\big)+O(\lVert\Sigma\rVert^{3/2})
\;=\;\frac{\tau^2}{2}\sum_{i=1}^d\frac{\lambda_i}{1+\gamma\lambda_i}\;+\;O(\lVert\Sigma\rVert^{3/2}),
\label{eq:derand}
\end{equation}
and the sum is bounded by $\dfrac{1}{\gamma}\cdot\dfrac12\sum_i\dfrac{\gamma\lambda_i}{1+\gamma\lambda_i}\le\dfrac{d}{2\gamma}$.
Consequently the train and test losses of $\pi_Q$ and of $\theta$ each differ by at most
$\tfrac12\tau^2\sum_i\lambda_i/(1+\gamma\lambda_i)\le\tau^2 d/(2\gamma)$, so the two
generalization gaps differ by at most $\tau^2\sum_i\lambda_i/(1+\gamma\lambda_i)\le\tau^2 d/\gamma$.
\end{lemma}
\begin{proof}
Since $\mathbb{E}_{\vartheta\sim Q}[\vartheta-\theta]=0$, a second-order Taylor expansion of
$L$ about $\theta$ gives
$\mathbb{E}_Q L(\vartheta)-L(\theta)=\tfrac12\,\mathbb{E}_Q[(\vartheta-\theta)^\top\nabla^2 L\,(\vartheta-\theta)]+O(\lVert\Sigma\rVert^{3/2})
=\tfrac12\mathrm{tr}(\Sigma\nabla^2 L)+O(\lVert\Sigma\rVert^{3/2})$, the remainder controlled by
the third derivative of the smooth PQC expectation and the Gaussian third moment. For the
logistic head, $\nabla^2 L=\mathbb{E}_s[\alpha^2 p_s(1{-}p_s)\nabla f_s\nabla f_s^\top]
+\mathbb{E}_s[\alpha(p_s{-}y_s)\nabla^2 f_s]=\Fmat+\mathcal{R}$, and $\mathcal{R}$ vanishes as the
fit improves ($p_s\!\to\!y_s$). Substituting $\Sigma=\tau^2(I_d+\gamma\Fmat)^{-1}$ and
diagonalizing $\Fmat$ gives $\mathrm{tr}(\Sigma\Fmat)=\tau^2\sum_i\lambda_i/(1+\gamma\lambda_i)$,
which is~\eqref{eq:derand}. Writing $\lambda_i/(1+\gamma\lambda_i)=\gamma^{-1}\cdot\gamma\lambda_i/(1+\gamma\lambda_i)$
and $\gamma\lambda_i/(1+\gamma\lambda_i)\in[0,1)$ gives the $d/2\gamma$ envelope. A flat
direction ($\lambda_i\!\to\!0$) contributes $\lambda_i/(1+\gamma\lambda_i)\!\to\!\lambda_i\!\to\!0$
despite its posterior variance $\Sigma_{ii}\!\to\!\tau^2$: the variance is large but the
curvature that converts it into a loss change is exactly $\lambda_i$, so flat directions are
harmless. The gap statement follows by applying the bound to the train and test losses
separately and using the triangle inequality.
\end{proof}

Note that the derandomization cost $\tfrac12\tau^2\sum_i\lambda_i/(1+\gamma\lambda_i)$ is
exactly $1/\gamma$ times the KL curvature term $\tfrac12\sum_i\gamma\lambda_i/(1+\gamma\lambda_i)$
that we dropped in Eq.~\eqref{eq:klfisher}: the same quantity we discarded to expose $\Deff$
reappears, scaled by $1/\gamma$, as the price of derandomization, and $\gamma{=}50$ makes it
small ($\le d/100$ for $\tau^2{=}1$). This ties the two approximations together and replaces
the informal ``low-temperature'' argument in the main text.

\subsection{Proof of Proposition~\ref{prop:ent}}
Write the readout expectation in the Heisenberg picture,
$f(s,\theta)=\langle\psi(s)|\,U(\theta)^\dagger Z_0\, U(\theta)\,|\psi(s)\rangle$, and let
gate $i$ be a Pauli rotation $e^{-\mathrm{i}\theta_i G_i/2}$ with $G_i$ a single-qubit
Pauli on wire $q(i)$. Splitting the circuit at gate $i$ as $U=U_{>i}\,e^{-\mathrm{i}\theta_i G_i/2}\,U_{\le i}$
and defining the back-propagated observable $O_i(\theta)=U_{>i}^\dagger Z_0\,U_{>i}$, the
parameter-shift identity gives
\begin{equation}
\partial_{\theta_i} f(s,\theta)
=-\tfrac{\mathrm i}{2}\,\langle\phi_i(s)|\,[\,G_i,\,O_i(\theta)\,]\,|\phi_i(s)\rangle,
\qquad |\phi_i(s)\rangle=U_{\le i}|\psi(s)\rangle .
\label{eq:pshift}
\end{equation}
The (output-)Fisher matrix is $\Fmat_{ij}=\mathbb{E}_s[\,c(s)\,\partial_{\theta_i} f\,\partial_{\theta_j} f\,]$
with $c(s)=\mathrm{sc}^2 p(1{-}p)>0$, so $\Fmat=\mathbb{E}_s[c(s)\,g_s g_s^\top]$ where
$g_s=\nabla_\theta f(s,\theta)$; in particular $\mathrm{rank}(\Fmat)\le|\{i:\ \partial_{\theta_i} f\not\equiv 0\}|$,
with equality generically.

\begin{definition}[Backward light-cone]
\label{def:cone}
Let $\mathcal{L}(\theta)\subseteq\{1,\dots,d\}$ be the set of parameters $i$ for which
$O_i(\theta)$ has non-trivial support on wire $q(i)$, i.e.\ $[G_i,O_i(\theta)]\neq 0$.
\end{definition}

\begin{lemma}[Entanglement enlarges the light-cone]
\label{lem:lightcone}
Fix the placement of all rotation gates. If a set of two-qubit entangling gates is added
to the circuit (keeping all rotations), then $\mathcal{L}(\theta)$ can only grow:
$\mathcal{L}_{\mathrm{without}}(\theta)\subseteq\mathcal{L}_{\mathrm{with}}(\theta)$ for
generic $\theta$. In particular, with no entangling gates $O_i(\theta)$ is supported on
wire $0$ alone, so $[G_i,O_i]=0$ for every rotation on a wire $q(i)\neq 0$ and
$\mathcal{L}=\{i:q(i)=0\}$.
\end{lemma}
\begin{proof}
$O_i(\theta)=U_{>i}^\dagger Z_0 U_{>i}$ is $Z_0$ conjugated by the gates after $i$. Each
single-qubit rotation preserves the support of an operator on its own wire; each two-qubit
gate can only enlarge support to the wires it couples. Hence the support of $O_i$ is
contained in the causal cone of $Z_0$ under the gates in $U_{>i}$, which is monotone
non-decreasing under the addition of entangling gates. With no entangling gates the cone is
$\{0\}$, so $O_i$ commutes with any $G_i$ on a different wire, giving $\partial_{\theta_i}f\equiv0$
by \eqref{eq:pshift}. Adding entangling gates that reach wire $q(i)$ makes $O_i$ act
non-trivially there, so $[G_i,O_i]\neq0$ for generic $\theta$ and $i$ enters
$\mathcal{L}$.
\end{proof}

\begin{assumption}[Genericity]
\label{ass:generic}
For generic parameters the nonzero gradients $\{g_s\}$ span
$\mathrm{span}\{e_i:i\in\mathcal{L}(\theta)\}$ and the induced Fisher eigenvalues on this
subspace are of comparable magnitude (no exact degeneracies or cancellations).
\end{assumption}

\begin{proof}[Proof of Proposition~\ref{prop:ent}]
The rank identity $\mathrm{rank}(\Fmat)=|\{i:\partial_{\theta_i}f\not\equiv0\}|=|\mathcal{L}(\theta)|$
holds because $\Fmat=\mathbb{E}_s[c(s)g_sg_s^\top]$ is supported on
$\mathrm{span}\{e_i:i\in\mathcal{L}(\theta)\}$ (parameters outside the light-cone have
$\partial_{\theta_i}f\equiv0$ by Lemma~\ref{lem:lightcone} and~\eqref{eq:pshift}), and by
Assumption~\ref{ass:generic} the in-cone gradients span that subspace, so the inequality is
an equality. Monotonicity of $|\mathcal{L}(\theta)|$ in the entangling connectivity is
Lemma~\ref{lem:lightcone}. Because $\Deff(\gamma)=\sum_i\log(1+\gamma\lambda_i)$ receives a
contribution only from the $\mathrm{rank}(\Fmat)$ nonzero eigenvalues and
$\Deff(\gamma)=\mathrm{rank}(\Fmat)\log\gamma+O(1)$ for large $\gamma$, any increase of
$\mathrm{rank}(\Fmat)$ raises $\Deff(\gamma)$ at fixed $d$. Hence increasing entangling
connectivity raises $\Deff$ while $d$ is fixed.
\end{proof}
The empirical counterpart is Table~\ref{tab:appfisher}: at fixed depth (fixed $d$) the
measured $\Deff=\log\det(I+\gamma\Fmat)$ rises monotonically from its no-entanglement value
(none $\to$ linear $\to$ full), tracking the rank of the Fisher matrix.

\paragraph{Scope, a worked example, and failure modes.} We state precisely what
Proposition~\ref{prop:ent} does and does not give. \emph{Scope.} It applies to
hardware-efficient ans\"atze in which (a) each trainable gate is a Pauli rotation, (b) the
readout is a fixed Pauli observable, and (c) encoding gates precede the trainable block. The
two claims have different strengths: \emph{rank monotonicity}
($\mathrm{rank}(\Fmat_{\mathrm{with}})\ge\mathrm{rank}(\Fmat_{\mathrm{without}})$) is
rigorous---it is Lemma~\ref{lem:lightcone} plus the support argument, with no genericity
needed for the inequality (genericity is used only to turn it into the equality
$\mathrm{rank}=|\mathcal{L}|$). The stronger \emph{effective-dimension monotonicity}
$\Deff(\gamma)_{\mathrm{with}}\ge\Deff(\gamma)_{\mathrm{without}}$ holds unconditionally only
as $\gamma\to\infty$ (where $\Deff\!\to\!\mathrm{rank}\cdot\log\gamma$); at finite $\gamma$ a
rank increase raises $\Deff$ provided the new eigenvalues are not vanishingly small, since
$\Deff=\sum_i\log(1+\gamma\lambda_i)$ and a mode with $\lambda\!\approx\!0$ contributes
$\approx0$. We therefore separate the two in the statement.

\emph{Worked example.} Take $n{=}2$ qubits, one layer, readout $Z_0$, trainable
$R_Y(\theta_0)$ on qubit $0$ and $R_Y(\theta_1)$ on qubit $1$. Without a CNOT,
$\partial_{\theta_1}\langle Z_0\rangle\equiv0$ (the back-propagated observable is $Z_0\!\otimes\!I$,
which commutes with $Y_1$), so $\Fmat=\mathrm{diag}(c\,\langle\dots\rangle,0)$ has rank $1$.
Inserting $\mathrm{CNOT}_{0\to1}$ before the readout maps $Z_0\mapsto Z_0$ but
$\mathrm{CNOT}_{1\to0}$ maps $Z_0\mapsto Z_0Z_1$, which no longer commutes with $Y_1$; then
$\partial_{\theta_1}\langle Z_0\rangle=-\tfrac{\mathrm i}{2}\langle[Y_1,Z_0Z_1]\rangle\neq0$
generically and $\mathrm{rank}(\Fmat)=2$. This is the mechanism in closed form.

\emph{Failure modes.} The rank does \emph{not} increase when the added entanglement is
neutralized by structure: (i) a symmetry or commuting-gate configuration for which
$[G_i,O_i(\theta)]=0$ despite non-trivial support (e.g.\ $O_i$ proportional to $Z_{q(i)}$ and
$G_i=Z_{q(i)}$); (ii) fine-tuned parameters at which the new gradients vanish (a measure-zero
set excluded by genericity); and (iii) the barren-plateau regime, where all
$\lambda_i\!\to\!0$ exponentially, so although the \emph{rank} may rise the finite-$\gamma$
$\Deff$ stays near zero---entanglement then inflates capacity only once the circuit is
trained away from the plateau. These are exactly the cases our experiments avoid
(shallow circuits, generic trained parameters), and they delimit the claim.

\section{Additional experimental detail}
\label{app:extra}

\paragraph{Setup.} All PQCs use $n{=}4$ qubits simulated with PennyLane's
\texttt{default.qubit} and a \emph{single} readout observable $\langle Z_0\rangle$. Inputs
are angle-encoded with $R_Y$; the ansatz is $L$ layers of $R_Y,R_Z$ per qubit ($d{=}2nL$
trainable angles) followed by a fixed CNOT pattern (\texttt{none}/\texttt{linear}/\texttt{full}).
Figure~\ref{fig:circuit} shows the three ans\"atze: they share the identical rotation layout
(hence $d$) and differ only in the entangling block.

\begin{figure}[h]
\centering
\begin{tabular}{@{}l@{\hskip 1em}l@{}}
\raisebox{1.1\height}{\small\texttt{none}} &
\begin{quantikz}[column sep=4pt,row sep=5pt]
\lstick{$q_0$} & \gate{R_Y} & \gate{R_Y} & \gate{R_Z} & \meter{Z_0}\\
\lstick{$q_1$} & \gate{R_Y} & \gate{R_Y} & \gate{R_Z} & \qw\\
\lstick{$q_2$} & \gate{R_Y} & \gate{R_Y} & \gate{R_Z} & \qw\\
\lstick{$q_3$} & \gate{R_Y} & \gate{R_Y} & \gate{R_Z} & \qw
\end{quantikz}\\[4pt]
\raisebox{1.1\height}{\small\texttt{linear}} &
\begin{quantikz}[column sep=4pt,row sep=5pt]
\lstick{$q_0$} & \gate{R_Y} & \gate{R_Y} & \gate{R_Z} & \ctrl{1} & \qw & \qw & \meter{Z_0}\\
\lstick{$q_1$} & \gate{R_Y} & \gate{R_Y} & \gate{R_Z} & \targ{} & \ctrl{1} & \qw & \qw\\
\lstick{$q_2$} & \gate{R_Y} & \gate{R_Y} & \gate{R_Z} & \qw & \targ{} & \ctrl{1} & \qw\\
\lstick{$q_3$} & \gate{R_Y} & \gate{R_Y} & \gate{R_Z} & \qw & \qw & \targ{} & \qw
\end{quantikz}\\[4pt]
\raisebox{1.1\height}{\small\texttt{full}} &
\begin{quantikz}[column sep=4pt,row sep=5pt]
\lstick{$q_0$} & \gate{R_Y} & \gate{R_Y} & \gate{R_Z} & \ctrl{1} & \ctrl{2} & \ctrl{3} & \qw & \qw & \qw & \meter{Z_0}\\
\lstick{$q_1$} & \gate{R_Y} & \gate{R_Y} & \gate{R_Z} & \targ{} & \qw & \qw & \ctrl{1} & \ctrl{2} & \qw & \qw\\
\lstick{$q_2$} & \gate{R_Y} & \gate{R_Y} & \gate{R_Z} & \qw & \targ{} & \qw & \targ{} & \qw & \ctrl{1} & \qw\\
\lstick{$q_3$} & \gate{R_Y} & \gate{R_Y} & \gate{R_Z} & \qw & \qw & \targ{} & \qw & \targ{} & \targ{} & \qw
\end{quantikz}
\end{tabular}
\caption{The three ans\"atze at $n{=}4$ (one layer of the repeated block shown; the block
repeats $L$ times, here between the encoding $R_Y(x_i)$ and the $\langle Z_0\rangle$ readout).
All three use the \emph{same} $R_Y,R_Z$ rotations---hence the same parameter count
$d{=}2nL$---and differ only in the entangling block: none (no CNOTs), a linear CNOT chain, or
all-to-all CNOTs. This is the ``fixed parameter count, vary only entanglement'' comparison.}
\label{fig:circuit}
\end{figure}
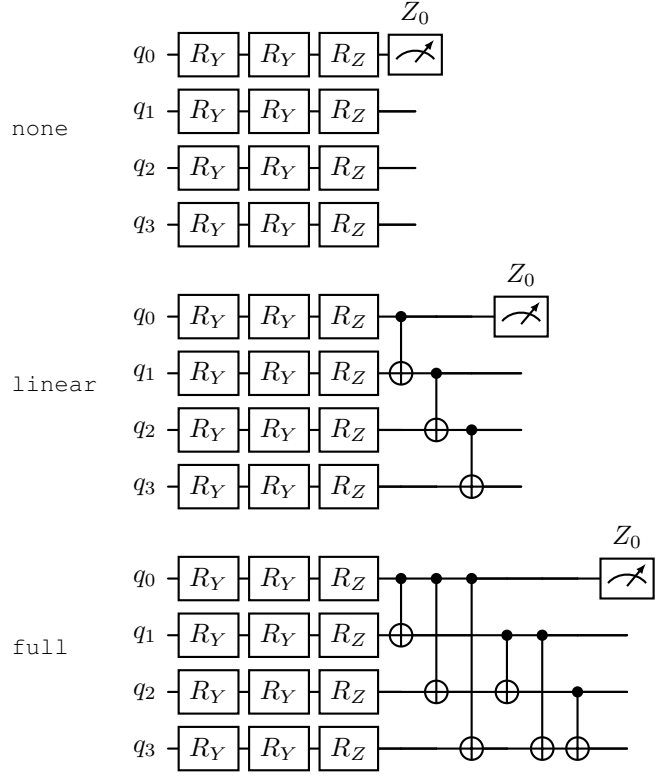
\emph{Supervised models} (Sections~\ref{sec:decouple}--\ref{sec:supervised}) are binary
classifiers $\hat y=\mathrm{sign}(\langle Z_0\rangle)$ trained by Adam ($600$ steps, full
batch, lr $0.08$) with the cross-entropy loss on the interaction target
$y=\mathrm{sign}(x_1x_2-x_3x_4+\tfrac12 x_1)$, $x\sim U[-1,1]^4$, and evaluated on $2000$
held-out inputs. \emph{RL policies} (Section~\ref{sec:rl}) are contextual-bandit policies
$\pi(a{=}1\mid x)=\sigma(\alpha\langle Z_0\rangle)$ over contexts $x\sim U[-1,1]^4$, trained
by REINFORCE ($600$ updates, lr $0.08$) on bandit feedback (reward $1$ iff the sampled
action equals $a^\star(x)=\mathbb{1}[x_1x_2-x_3x_4+\tfrac12 x_1>0]$, the correct label is
never revealed); the gap is training reward minus reward on $2000$ held-out contexts. Each
cell averages $8$ (supervised) or $16$ (bandit) seeds. The Meyer--Wallach $Q$ and the
Fisher matrix are evaluated at the trained parameters over held-out inputs.

\paragraph{Fisher effective dimension: exact computation.} To remove ambiguity we specify
the estimator precisely. (i) We use the \emph{empirical output Fisher} of the policy's
Bernoulli head, not the parameter/loss Fisher and not the exact Fisher:
$\Fmat(\theta)=\tfrac1M\sum_{x} c(x)\,g_x g_x^\top$ with $g_x=\nabla_\theta\langle Z_0\rangle_{x,\theta}$
and Bernoulli weight $c(x)=\alpha^2\,p(x)(1{-}p(x))$, $p=\sigma(\alpha\langle Z_0\rangle)$.
(ii) It is evaluated at the \emph{trained} parameters $\theta$ over $M{=}40$ \emph{held-out}
inputs drawn from the input distribution (not the training set). (iii) The per-input
gradients $g_x$ are computed by autograd (exact parameter-shift). (iv) Eigenvalues
$\lambda_i$ of $\Fmat$ are used \emph{unnormalized} (no trace rescaling). (v) The effective
dimension is $\Deff=\log\det(I+\gamma\Fmat)=\sum_i\log(1+\gamma\lambda_i)$ with the constant
$\gamma=50$ used identically everywhere (Remark~\ref{rem:gamma}). (vi) For a reported cell we
compute $\Deff$ per seed and then \emph{average} $\Deff$ across seeds (we do not pool the
Fisher matrices). The participation ratio $(\sum_i\lambda_i)^2/\sum_i\lambda_i^2$, when
reported, is computed from the same spectra. \emph{Why $\gamma{=}50$, and does it matter.}
The value $50$ places the observed spectra (eigenvalues $\lambda\!\in\![10^{-2},2]$) in the
informative band where $\gamma\lambda$ spans $\approx0.5$--$100$, so $\Deff$ is neither
saturated nor collapsed to zero. Crucially the conclusions we draw from $\Deff$ are insensitive to this choice, since
$\Deff(\gamma)=\sum_i\log(1+\gamma\lambda_i)$ is monotone in the Fisher spectrum. We
recomputed $\Deff$ for all $18$ configurations of Table~\ref{tab:fisher} at
$\gamma\in\{10,25,50,100,200\}$: the Spearman correlation $\rho(\Deff,\mathrm{gap})$ is
$0.74,0.70,0.67,0.67,0.64$ respectively---stable in sign and magnitude---and the ordering
none $<$ linear $<$ full holds at every $\gamma$ (e.g.\ at $L{=}2$, $\Deff$ goes
$2.5{\to}4.8{\to}11.5$ at $\gamma{=}10$ and $5.8{\to}12.2{\to}25.6$ at $\gamma{=}200$). We fix
a single $\gamma=50$ only so that $\Deff$ is numerically comparable across rows.

\paragraph{$\Deff$ dominates the learned norm---a powered, bootstrapped test.} A comparison
between two correlations of $0.76$ and $0.70$ over $18$ points would be underpowered, so we
built a much larger and more varied grid: $5$ entangling patterns
(\texttt{none}/\texttt{linear}/\texttt{ring}/\texttt{alt}/\texttt{full}) $\times$ depths
$L\!\in\!\{1,\dots,5\}$ $\times$ $N\!\in\!\{16,32,64\}$ $\times$ $4$ seeds $=300$
configurations, which dissociates $\Deff$ from the learned norm and from parameter count.
Across these $300$ points the Spearman correlations with the gap are $\rho(\Deff)=0.82$,
$\rho(\lVert\theta{-}\theta_0\rVert^2)=0.73$, $\rho(Q)=0.72$, $\rho(\text{params})=0.45$.
A $5000$-resample bootstrap over configurations gives
$\Delta\rho=\rho(\Deff)-\rho(\text{norm})=0.090$ with $95\%$ CI $[0.041,0.144]$ (bootstrap
$P(\Delta\rho{>}0){=}1.00$), and $\rho(\Deff)-\rho(Q)=0.100$, CI $[0.046,0.156]$: the Fisher
effective dimension is a \emph{significantly} better predictor of the gap than the learned
norm or the entangling power. A multiple regression on the same grid loads almost entirely on
$\Deff$ (standardized coefficient $+0.16$ vs.\ $\le0.06$ in magnitude for the others), and
adding the entangling pattern on top of $\Deff$ raises $R^2$ by $<0.01$ while adding $\Deff$
on top of the pattern raises it by $0.14$. Entanglement therefore acts on the gap
\emph{through} $\Deff$, which screens it off---the precise, and now powered, sense in which
$\Deff$ is the governing quantity.

\paragraph{Fisher effective dimension and gap, per configuration.}
Table~\ref{tab:appfisher} gives the Fisher effective dimension
$\Deff=\log\det(I+\gamma\Fmat)$ (Eq.~\eqref{eq:deff}, fixed $\gamma=50$) and the
gap for every connectivity at depths $L{=}2,3$ and $N_{\mathrm{train}}{=}16$. At
fixed depth (fixed $d$), $\Deff$ increases monotonically with entangling connectivity, and
so does the gap.

\begin{table}[h]
\centering
\caption{Fisher effective dimension $\Deff=\log\det(I+\gamma\Fmat)$ and generalization gap.
At each fixed $(L,N_{\mathrm{train}})$ block the parameter count $d=8L$ is identical across
the three rows; only the entangling connectivity---and hence $\Deff$ and the gap---changes.
$\Deff$ rises monotonically none $\to$ linear $\to$ full.}
\label{tab:appfisher}
\begin{tabular}{llccc}
\toprule
depth & connectivity & $N_{\mathrm{train}}$ & Fisher $\Deff$ (log-det) & gap \\
\midrule
$L{=}2$ & \texttt{none}   & $16$ & $3.3$  & $0.101$ \\
$L{=}2$ & \texttt{linear} & $16$ & $9.3$  & $0.146$ \\
$L{=}2$ & \texttt{full}   & $16$ & $21.8$ & $0.381$ \\
\midrule
$L{=}3$ & \texttt{none}   & $16$ & $4.0$  & $0.101$ \\
$L{=}3$ & \texttt{linear} & $16$ & $29.5$ & $0.440$ \\
$L{=}3$ & \texttt{full}   & $16$ & $36.5$ & $0.431$ \\
\bottomrule
\end{tabular}
\end{table}

The ordering none $<$ linear $<$ full holds for both $\Deff$ and the gap at $L{=}2$ and
$L{=}3$; a single entangling layer ($L{=}1$) leaves the gap unchanged, the depth-gating
referred to in the main text.

\begin{figure}[h]
\centering
\includegraphics[width=0.5\textwidth]{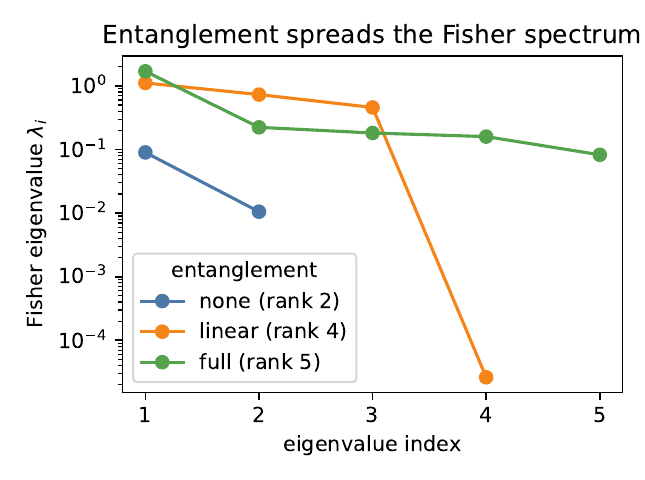}
\caption{Fisher eigenvalue spectrum of the trained circuits at $L{=}2$ (single-observable
readout, log scale). The non-entangled circuit has effectively two non-zero eigenvalues,
\texttt{linear} three, and \texttt{full} five or more: entanglement raises the rank of the
Fisher matrix and hence the effective dimension $\Deff=\log\det(I+\gamma\Fmat)$, exactly the
quantity charged by the bound. This is the spectral content of Proposition~\ref{prop:ent}.}
\label{fig:spectrum}
\end{figure}

\paragraph{Multi-step value-function generalization and controls.} The value-function
experiment (Section~\ref{sec:rl}) fits a single-observable VQC to
$V^\star(x)=\sum_{t=0}^{5}0.9^{t}r(x_t)$ under the deterministic transition
$x'=\tanh(1.3\,\mathrm{roll}(x)+0.2x)$ with $r$ the interaction rule, by Adam regression
($800$ steps) on $N$ states, reporting held-out $R^2$ over $1000$ states. The controls of
Section~\ref{sec:controls} use $N{=}64$: matched-accuracy early-stops at train accuracy
$0.72$; the readout control measures $Z_1$; the optimizer control uses SGD (lr $0.5$). The
hardware run trains the $N{=}16$ classifiers in simulation and evaluates $\langle Z_0\rangle$
for all $56$ inputs on \texttt{ibm\_aachen} via \texttt{EstimatorV2} at $4096$ shots,
optimization level $1$.

\paragraph{Hardware credibility: calibration, noise, and shot statistics.} The device is IBM
Heron \texttt{ibm\_aachen} ($156$ qubits). At the time of execution its median two-qubit gate
error was $1.8\times10^{-3}$ (typical range $7\times10^{-4}$ to a few $\times10^{-3}$, excluding
one inoperable pair) and median $T_1\approx218\,\mu$s. Readout error mitigation was
\emph{not} applied---the effect we report survives raw device noise. Three cross-checks
support the hardware numbers (Figure~\ref{fig:controls}, right). (i) A \emph{noise-model
simulator} built from the \texttt{ibm\_aachen} calibration (\texttt{AerSimulator} with
\texttt{NoiseModel.from\_backend}) reproduces the ordering
(gap $0.14/0.33/0.40$ for none/linear/full), i.e.\ the modelled noise does not remove the
effect. (ii) \emph{Shot noise} is small: at $4096$ shots each $\langle Z_0\rangle$ estimate
has standard error $\approx\!1/\sqrt{4096}\approx0.016$, which flips a sign only for inputs
within that margin of the decision boundary, so the accuracy-based gap is stable. (iii) The
raw hardware gaps ($0.14/0.33/0.31$) agree with both the noiseless and noise-model values to
within $\approx0.1$, and preserve the none $<$ \{linear, full\} ordering. We report the
hardware result as corroboration of the simulation evidence, not as a separate quantitative
claim.

\paragraph{Compute and reproducibility.} All classical simulations use PennyLane
\texttt{default.qubit} with the PyTorch interface (analytic gradients). Each supervised cell
averages $8$ seeds, each bandit/robustness/real-data cell $16$/$8$/$6$ seeds, and the
value-function and policy-gradient cells $8$/$10$ seeds; error bars are one standard error.
Targets: interaction rule $\mathrm{sign}(x_1x_2{-}x_3x_4{+}\tfrac12 x_1)$ (main);
$4$-bit parity $\mathrm{sign}(\prod_i \mathrm{sign}(x_i))$; random Fourier
$\mathrm{sign}(\sum_{k=1}^{6}\cos(w_k\!\cdot\!x+b_k))$ with fixed $w_k,b_k$; entangled-teacher
labels $\mathrm{sign}(\langle Z_0\rangle)$ of a fixed $2$-layer \texttt{StronglyEntangling}
teacher. Real datasets are the scikit-learn Iris (versicolor vs.\ virginica), Breast Cancer
Wisconsin, and Wine (classes $0$/$1$) sets, standardized, reduced to four features by PCA when
needed, and min--max scaled to $[-1,1]$.

\end{document}